\RequirePackage{tikz}
\documentclass[sn-apa]{sn-jnl}

\let\B\relax
\let\T\relax

\usepackage{amssymb}
\usepackage{amsmath}
\usepackage{multirow}
\usepackage{cleveref}
\usepackage{bm}
\usepackage[T1]{fontenc}
\usepackage{polski}
\usepackage[english]{babel}
\usepackage[inline]{enumitem}
\usepackage{tikz}
\usepackage{mathbbol}
\usepackage{mathrsfs}
\usepackage{dutchcal}
\usepackage{longtable}
\usepackage{caption}
\usepackage{subcaption}

\renewcommand{\imath}{\iota}

\newcommand{\T}{\mathcal{T}}
\newcommand{\F}{\mathcal{F}}

\newcommand{\M}{\mathcal{M}}
\newcommand{\D}{\mathcal{D}}

\newcommand{\I}{\mathcal{I}}

\newcommand{\Tab}{\mathfrak{T}}
\newcommand{\B}{\mathfrak{B}}
\newcommand{\HFLK}{\mathsf{HFL_K}}
\newcommand{\TCHFLK}{\mathscr{TC}(\FOHLLDK)}
\newcommand{\HL}{\mathsf{HL}}
\newcommand{\FOML}{\mathsf{FOML}}
\newcommand{\FOHL}{\mathsf{FOHL}}
\newcommand{\FOHLLD}{\mathsf{FOHL}_{\lambda,\imath}}
\newcommand{\FOHTL}{\mathsf{FOHL}^\mathsf{F,P}_{\lambda,\imath}}
\newcommand{\FOHLLDK}{\FOHTL}

\newcommand{\ass}{\mathcal{v}}
\newcommand{\future}{\mathsf{F}}
\newcommand{\past}{\mathsf{P}}
\newcommand{\going}{\mathsf{G}}

\newenvironment{proof}{\noindent\textit{Proof}.\ }{\hfill$\dashv$}

\makeatletter
\newtheorem*{rep@theorem}{\rep@title}
\newcommand{\newreptheorem}[2]{%
	\newenvironment{rep#1}[1]{%
		\def\rep@title{#2 \ref{##1}}%
		\begin{rep@theorem}}%
		{\end{rep@theorem}}}
\makeatother
\jyear{2022}%

\theoremstyle{thmstyleone}%
\newtheorem{theorem}{Theorem}
\newtheorem{lemma}[theorem]{Lemma}%
\newtheorem{proposition}[theorem]{Proposition}%

\theoremstyle{thmstyletwo}%
\newtheorem{example}{Example}%
\newtheorem{fact}{Fact}

\theoremstyle{thmstylethree}%

\begin{document}
	
	\title{Definite descriptions and hybrid tense logic}
	
	
	\author{\fnm{Andrzej} \sur{Indrzejczak} \tanm{ORCID: 0000-0003-4063-1651}}\email{andrzej.indrzejczak@filhist.uni.lodz.pl}
	\equalcont{These authors contributed equally to this work.}
	
	\author*{\fnm{Micha{\l}} \sur{Zawidzki} \tanm{ORCID: 0000-0002-2394-6056}}\email{michal.zawidzki@filhist.uni.lodz.pl}
	\equalcont{These authors contributed equally to this work.}

	\affil{\orgdiv{Department of Logic}, \orgname{University of Lodz}, \orgaddress{\street{Lindleya 3/5}, \city{\L\'od\'z}, \postcode{90-131}, \country{Poland}}}

	
	\abstract{ We provide a version of first-order hybrid tense logic with predicate abstracts and definite descriptions as the only non-rigid terms. It is formalised by means of a tableau calculus working on sat-formulas.
		A particular theory of DD exploited here is essentially based on the approach of Russell,
		but with descriptions treated as genuine terms. However, the reductionist aspect of the Russellian approach is retained in several ways. 
		Moreover, a special form of tense definite descriptions is formally developed.
		A constructive proof of the interpolation theorem for this calculus is given, which is an extension of the result provided by Blackburn and Marx.}

	
	
	
	\maketitle
	
%

\section{Introduction}
\label{sec::Introduction}

Hybrid logic ($\HL$) is an important augmentation of standard modal logic with rich syntactic resources.  
The basic language of $\HL$ is obtained by adding a second sort of propositional atoms, called \emph{nominals}, each of which holds true at exactly one state of a model and serves as a name of this state. Additionally, one can introduce several
extra operators; the most important one is the \emph{satisfaction}, or $@$-, operator which takes as its first argument a nominal $\bm{j}$ and as the second one an arbitrary $\HL$-formula $\varphi$. A formula $@_{\bm{j}} \varphi$ indicates that $\varphi$ is satisfied at the state denoted by $\bm{j}$. This allows us to internalise an essential part of the semantics in the language. Another specific operator is the \emph{downarrow} binder ($\downarrow$) which binds the value of a state variable to the current state. What is nice about $\HL$ is that the additional hybrid machinery does not seriously affect the modal logic core it is based on. In particular, modifications in the relational semantics are minimal. The concept of \emph{frame} remains intact. Only at the level of models we have some changes. Moreover, adding a binder-free hybrid toolkit typically does not increase the computational complexity of the underlying modal logic.
These relatively small modifications of standard modal languages give us many advantages:
\begin{enumerate*}
	\item a more expressive language,
	\item a better behaviour in completeness theory,
	\item a more natural and simpler proof theory.\end{enumerate*}
In particular, defining frame conditions such as irreflexivity, asymmetry, trichotomy, and others, impossible in standard modal languages, becomes possible in $\HL$. 
This machinery and results are easily extendable to multimodal logics, in particular to tense and temporal logic~\citep{BlTz99,BlJo2012}. Proof theory of $\HL$ offers an even more general approach than applying labels popular in proof theory for standard modal logic, namely it allows for \emph{internalising} those labels as part of standard hybrid formulas~\citep{Brauner2011,Indrzejczak2010}.

$\HL$ offers considerable benefits pertaining to the interpolation property. It is well known that for many modal logics in standard languages this property fails. The situation is particularly bad for the first-order case; Fine (\citeyear{Fine79}) showed that the first-order variant of $\sf S5$ does not enjoy interpolation, and also all modal logics from the modal cube with constant domains fail to satisfy it. On the other hand, $\HL$ offers resources which significantly improve the situation. In this case the $\downarrow$-binder turns out particularly useful. The  uniform interpolation theorem for all propositional modal logics complete with respect to any class of frames definable in the bounded fragment of first-order logic was proved by Areces, Blackburn, and Marx~(\citeyear{ArBlMar01}). In the follow-up paper (\citeyear{ArBlMar03}) the result was extended to first-order hybrid logic ($\FOHL$). In both cases the results were obtained semantically and non-constructively, however, in the later work~\citep{BlMar03} a constructive proof of interpolation was also provided for a tableau calculus for $\FOHL$.

In this paper we provide an extension of the aforementioned tableau calculus and the interpolation theorem for a richer version of $\FOHL$ involving predicate abstracts and definite descriptions. Let us briefly comment on these two kinds of extensions.
Adding definite descriptions or other complex terms to $\FOHL$ increases the expressive power of the language, which has recently also been noticed in the area of description logics~\citep{ArtaleEtAl2021}. On the other hand, in the previous versions of $\FOHL$ due to Blackburn and Marx~(\citeyear{BlMar03}) or Bra\"uner~(\citeyear{Brauner2011}) only simple non-rigid terms were used to represent descriptions, whereas involving the $\imath$-operator enables us to unfold rich informational contents of descriptions which is often indispensable in checking the correctness of an argument. Several formal systems with rules characterising definite descriptions were proposed by Orlandelli~(\citeyear{Orlandelli2021}), Fitting and Mendelsohn~(\citeyear{FitMen98}), or Indrzejczak and Zawidzki~(\citeyear{IndZaw,IndZaw2023}). A novelty of our approach in this paper involves also the introduction of a new, specifically temporal, category of definite descriptions which we call \emph{tensal definite descriptions}. Formally they also are treated by means of the $\imath$-operator but applied to tense variables to obtain the phrases uniquely characterising some time points, hence syntactically they behave like nominals and tense variables and may also be used as first arguments of the satisfaction operator. Intuitively, descriptions of this kind correspond to phrases such as `the wedding day of Anne and Alex', `the moment in which this accident took place', `the first year of the French Revolution', etc. Although it seems that in the general setting of modal logics the introduction of such descriptive nominal phrases is not always needed, in the case of temporal interpretation such an extension of the language is very important since these phrases are commonly used.\footnote{In fact even in the case of alethic modalities they can be applied to express phrases such as the Leibnizian `the best of possible worlds'.} What differs in the way tensal definite descriptions are used in natural language and in the formal setting specified below is that in the latter they are syntactically treated as sentences uniquely characterising some points in time, whereas in the former they are usually noun phrases. Moreover, as we will show later, they are characterisable by means of well-behaved rules and the interpolation theorem applies to this extended system.

In addition to descriptions of two kinds we enrich our system with predicate abstracts built by means of the $\lambda$-operator. Such devices were introduced to the studies on $\FOML$ by Thomason and Stalnaker (\citeyear{StalnakerThomason1968}) and then the technique was developed by Fitting (\citeyear{Fitting1975}). In the realm of modal logic it has mainly been used for taking control over scoping difficulties concerning modal operators, but also complex terms like definite descriptions. Such an approach was developed by Fitting and Mendelsohn~(\citeyear{FitMen98}), and independently formalised in the form of cut-free sequent calculi by Orlandelli~(\citeyear{Orlandelli2021}) and Indrzejczak~(\citeyear{Indrzejczak2020a}). Orlandelli uses labels and characterises definite descriptions by means of a ternary designation predicate. Indrzejczak applies hybrid logic and handles definite descriptions by means of intensional equality. It provides the first version of $\FOHL$ with descriptions and $\lambda$-terms ($\FOHLLD$). 

The system of $\FOHLLD$ presented here is different from the one due to Indrzejczak (\citeyear{Indrzejczak2020a}). The latter was designed with the aim of following closely the approach of Fitting and Mendelsohn~(\citeyear{FitMen98}), which was based on the Hintikka axiom. Here we provide an approach based on the Russellian theory of definite descriptions enriched with predicate abstracts and developed in the setting of classical logic by Indrzejczak and Zawidzki~(\citeyear{IndZaw2023}).
The specific features of the Russellian approach to definite descriptions, its drawbacks and advantages were discussed at length by Indrzejczak~(\citeyear{Indrzejczak2021c}), so we omit its presentation. It should be nevertheless stressed that in spite of the fact that Russell treated descriptions as incomplete symbols eliminable by means of contextual definitions, we treat them as genuine terms. However, the reductionist aspect of the Russellian approach is retained in several ways. At the level of syntax the occurrences of definite descriptions are restricted to arguments of predicate abstracts forming so-called \emph{$\lambda$-atoms}. At the level of calculus definite descriptions cannot be instantiated for variables in quantifier rules, but they are eliminated with special rules for $\lambda$-atoms.
Eventually, at the level of semantics definite descriptions are not defined by an interpretation function, but by satisfaction clauses for $\lambda$-atoms. Therefore, their semantic treatment is different than the one known from the Fitting and Mendelsohn approach. It leads to less complex proofs of completeness of the calculus and to different rules characterising definite descriptions which are simpler than the ones from the sequent calculus by Indrzejczak~(\citeyear{Indrzejczak2020a}). Hybridised versions of the rules for $\lambda$-atoms are added here to the tableau calculus by Blackburn and Marx (\citeyear{BlMar03}), which allows us to maximally shorten the proof of the Interpolation Theorem by referring to their rules for calculating interpolants.

In~\Cref{sec:Preliminaries} we briefly characterise the language and semantics of our logic. The tableau calculus and the completeness proof for it are presented in Sections~\ref{sec::Tableaux} and~\ref{sec::Completeness}. In \Cref{sec::Interpolation} we extend the proof of the Interpolation Theorem presented by Blackburn and Marx~(\citeyear{BlMar03}). We conclude the paper with a brief comparison of the present system with Indrzejczak's former system~(\citeyear{Indrzejczak2020a}) and with some open problems.

\section{Preliminaries}
\label{sec:Preliminaries}

In what follows we will provide a formal characterisation of first-order hybrid tense logic with definite descriptions, abbreviated as $\FOHTL$. The language of $\FOHTL$ consists of the sets of logical and non-logical expressions. The former is constituted by:
\begin{itemize}
	\item a countably infinite set of individual bound variables $\mathsf{BVAR} = \{x, y, z, \ldots\}$,
	\item a countably infinite set of individual free variables $\mathsf{FVAR} = \{a, b, c, \ldots\}$,
	\item a countably infinite set of tense variables $\mathsf{TVAR} = \{\bm{x},\bm{y},\bm{z},\ldots\}$,
	\item the identity predicate $=$,
	\item the (possibilistic) existential quantifier $\exists$,
	\item the abstraction operator $\lambda$,
	\item the definite description operator $\imath$,
	\item boolean connectives $\neg$, $\land$,
	\item tense operators $\mathsf{F}$ (somewhere in the future), $\mathsf{P}$ (somewhere in the past),
	\item the satisfaction operator $@$,
	\item the downarrow operator $\downarrow$.
\end{itemize}
The set of non-logical expressions of the language of $\FOHTL$ includes:
\begin{itemize}
	\item a countably infinite set of individual constants $\mathsf{CONS} = \{i, j, k, \ldots\}$,
	\item a countably infinite set of tense constants called nominals $\mathsf{NOM} = \{\bm{i}, \bm{j}, \bm{k},\ldots\}$,
	\item a countably infinite set of $n$-ary predicates $\mathsf{PRED}^n = \{P, Q, R, \ldots\}$, for each $n \in \mathbb{N}$. By $\mathsf{PRED}$ we will denote the union $\bigcup_{n=0}^\infty\mathsf{PRED}^n$.
\end{itemize}
Intuitively, nominals are introduced for naming time instances in the temporal domain of a model. Thus, on the one hand, they play a role of terms. On the other hand, however, at the level of syntax they are ordinary sentences. In particular, they can be combined by means of boolean and modal connectives. When a nominal $\bm{i}$ occurs independently in a sentence, its meaning can be read as ``the name of the current time instance is $\bm{i}$ (and thus, $\bm{i}$ holds only here)''. If it occurs in the scope of the satisfaction operator, it only serves as a name of the time instance it holds at. Tense-variables are double-faced expressions, too, which can serve both as labels of time instances and as full-fledged formulas, each being true at only one time instance. They can additionally be bound by the downarrow operator and by the iota-operator, but not by the quantifier or the lambda-operator. It is important to note that both nominals and the satisfaction operator are genuine language elements rather than an extra metalinguistic machinery. Observe that for convenience of notation we separate the sets of bound and free object variables. We do not do that for tense variables, as, with a slight violation of consistency, at the temporal level nominals often play an analogous role to free variables at the object level.

We will denote the set of well-formed terms, well-formed temporal formulas, and well-formed formulas of $\FOHTL$ by $\mathsf{TERM}$, $\mathsf{TFOR}$, and $\mathsf{FOR}$, respectively. The second set is only auxiliary and we introduce it to make the notation more uniform in the remainder of the section. All the sets are defined simultaneously by the following context-free grammars:
\begin{multline*}
	\hfill\mathsf{TERM} \ni \xi ::= a \mid i \mid \imath x \varphi[x],\hfill\\
	\hfill\mathsf{TFOR} \ni \bm{\xi} ::= \bm{x} \mid \bm{i} \mid \imath \bm{x} \varphi,\hfill\\
	\mathsf{FOR} \ni \varphi ::= \bot \mid P(\eta_1 ,\ldots, \eta_n) \mid \eta_1 = \eta_2 \mid \bm{\xi}\mid \neg \varphi \mid\hfill\\
	\hfill\varphi\land \varphi \mid\mathsf{F}\varphi \mid \mathsf{P}\varphi \mid \exists x \varphi \mid (\lambda x \varphi[x])(\xi) \mid @_{\bm{\xi}}\varphi \mid\ \downarrow_{\bm{x}}\! \varphi,
\end{multline*}
where $a \in \mathsf{FVAR}$, $i \in \mathsf{CONS}$, $x \in \mathsf{BVAR}$, $\varphi \in \mathsf{FOR}$, $\bm{a} \in \mathsf{TVAR}$, $\bm{i} \in \mathsf{NOM}$, $\bm{x} \in \mathsf{TVAR}$, $n\in\mathbb{N}$, $P \in \mathsf{PRED}^n$, $\eta_1, \ldots, \eta_n \in \mathsf{FVAR} \cup \mathsf{CONS}$, $\xi \in \mathsf{TERM}$, and $\bm{\xi} \in \mathsf{TFOR}$. We write $\varphi[\mathcal{x}]$ to indicate that $\mathcal{x}$ is free in $\varphi$. Observe that we require that in a definite description $\imath x\varphi$ a variable $x$ occurs freely in $\varphi$. Similarly, in an expression $\lambda x\varphi$ it is assumed that $x$ occurs freely in $\varphi$. On the other hand, for a temporal definite description $\imath\bm{x}\varphi$ we do not expect $\bm{x}$ to necessarily occur (freely) in $\varphi$. Note that since $\mathsf{BVAR} \cap \mathsf{FVAR} = \emptyset$, in a formula of the form $(\lambda x\varphi)(\xi)$, $x$ does not occur freely in $\xi$. Similarly, we require that in a formula of the form $\imath\bm{x}\varphi$ a tense variable $\bm{x}$ does not occur freely in $\varphi$.  For any $\eta_1,\eta_2$ in $\mathsf{FVAR}\cup\mathsf{CONS}$ or in $\mathsf{TVAR}\cup\mathsf{NOM}$, a formula
$\varphi[\eta_1/\eta_2]$ is the result of a uniform substitution of $\eta_2$ for $\eta_1$ in $\varphi$, whereas a formula $\varphi[\eta_1/\!/\eta_2]$ results from replacing some occurrences of $\eta_1$ with $\eta_2$ in $\varphi$.
Note that we can make substitutions and replacements only using variables or constants, but not definite descriptions. In practice, when constructing a tableau proof, variables are substituted only with free variables, however in the formulation of the semantics and in metalogical proofs it may happen that variables are substituted or replaced with bound variables. In such cases it is assumed that the variable substituting or replacing another variable in a formula is free after the substitution or replacement.

Let us now briefly discuss an informal reading of hybrid elements of $\mathsf{FOR}$. An expression $@_{\bm{\xi}} \varphi$, where $\bm{\xi} \in \mathsf{TFOR}$, reads ``$\varphi $ is satisfied at a time instance denoted by $\bm{\xi}$''. If $\bm{\xi}$ is of the form $\imath\bm{x}\varphi$, then $@_{\imath\bm{x}\varphi} \psi$ reads: ``$\psi$ holds at the only time instance at which $\varphi$ holds''. Expressions of the form $\imath\bm{x}\varphi$ play a double role which is similar to the one of nominals, that is, on the one hand, they unambiguously label time instances and on the other, they are formulas that hold at these time instances.  An expression $\downarrow_{\bm{x}}\varphi$ fixes the denotation of $\bm{x}$ to be the time instance the formula $\downarrow_{\bm{x}}\varphi$ is currently evaluated at.
Finally, we also use the following standard abbreviations:
\begin{gather*}
	\xi \neq \eta := \neg(\xi=\eta) \quad \top := \neg\bot \quad
	\forall x \varphi := \neg \exists x \neg\varphi \quad \mathsf{G} \varphi := \neg\mathsf{F}\neg\varphi \quad \mathsf{H} \varphi := \neg\mathsf{P}\neg\varphi\\
	\varphi \lor \psi := \neg(\neg \varphi \land \neg \psi) \quad
	\varphi \to \psi := \neg \varphi \lor \psi \quad \varphi \leftrightarrow \psi := (\varphi \to \psi) \land (\psi \to \varphi)
\end{gather*}

We define a \emph{tense first-order frame} as a tuple $\F = (\T, \prec, \D)$, where:
\begin{itemize}
	\item $\T$ is a non-empty set of \emph{time instances} (the \emph{universe} of $\F$),
	\item $\prec \subset \T \times \T$ is a \emph{relation of temporal precedence} on $\T$, and
	\item $\D$ is a non-empty set called an \emph{object domain}.
\end{itemize}
Given a frame $\F = (\T, \prec, \D)$, a \emph{tense first-order model} based on $\F$ is a pair $\M = (\F, \I)$, where $\I$ is an \emph{interpretation function} defined on $\mathsf{NOM} \cup \mathsf{CONS} \cup (\mathsf{PRED} \times \T)$ in the following way:
\begin{itemize}
	\item $\I(\bm{i}) \in \T$, for each $\bm{i} \in \mathsf{NOM}$,
	\item $\I(a) \in \D$, for each $a \in \mathsf{CONS}$,
	\item $\I(P, \bm{t}) \subseteq \overset{n}{\overbrace{\D\times\ldots\times\D}}$, for each $n \in \mathbb{N}$, and $P \in \mathsf{PRED}^n$.
\end{itemize}
Note that in our setting individual constants are rigidified, that is, they have the same interpretation at all time instances, whereas extensions of predicates may vary between different time instances. By making this choice we follow the approach of Blackburn and Marx~(\citeyear{BlMar03}).

Given a model $\M = ((\T, \prec,\D ), \I)$, an \emph{assignment} $\ass$ is a function defined on $\mathsf{TVAR}\cup\mathsf{FVAR}\cup\mathsf{BVAR}$ as follows:
\begin{itemize}
	\item $\ass(\bm{x}) \in \T$, for each $\bm{x} \in \mathsf{TVAR}$,
	\item $\ass(x) \in \D$, for each $x \in \mathsf{FVAR}\cup\mathsf{BVAR}$.
\end{itemize}
Moreover, for an assignment $\ass$, time instance $\bm{t} \in \T$, a variable $x \in \mathsf{FVAR}\cup\mathsf{BVAR}$ and an object $o \in \D$ we define an assignment $\ass[x\mapsto o]$ as:
\begin{center}
	$\ass[x\mapsto o](\mathcal{y}) = \begin{cases}
		o, &\text{if }\mathcal{y} = x,\\
		\ass(\mathcal{y}), &\text{otherwise}.
	\end{cases}$
\end{center}
Analogously, for a tense-variable $\bm{x}$ and time instance $\bm{t}$ we define the assignment $\ass[\bm{x}\mapsto\bm{t}]$ in the following way:
\begin{center}
	$\ass[\bm{x}\mapsto \bm{t}](\mathcal{y}) = \begin{cases}
		\bm{t}, &\text{if }\mathcal{y} = \bm{x}\\
		\ass(\mathcal{y}), &\text{otherwise}.
	\end{cases}$
\end{center}
Finally, for a model $\M = (\F, \I)$ and an assignment $\ass$ an \emph{interpretation $\I$ under $\ass$}, in short $\I_{\ass}$, is a function which coincides with $\I$ on $\mathsf{NOM} \cup \mathsf{CONS} \cup (\mathsf{PRED} \times \T)$ and with $\mathcal{v}$ on $\mathsf{TVAR}\cup\mathsf{FVAR}\cup\mathsf{BVAR}$. Henceforth, we will write $(\T, \prec,\D,\I)$ to denote the model $((\T, \prec,\D),\I)$.

Below, we inductively define the notion of \emph{satisfaction} of a formula $\varphi$ at a time instance $\bm{t}$ of a model $\M$ under an assignment $\ass$, in symbols $\M, \bm{t}, \ass \models \varphi$.
\noindent\begin{longtable}{rcl}
	$\M, \bm{t}, \ass \not\models \bot$ &&\\
	$\M, \bm{t}, \ass \models P(\eta_1,\ldots , \eta_n) $ & iff & $\langle \I_{\ass}(\eta_1), ..., \I_{\ass}(\eta_n) \rangle \in \I_\ass(P,\bm{t})$ \\ 
	$\M, \bm{t}, \ass  \models \eta_1 = \eta_2 $ & iff & $\I_{\ass}(\eta_1) = \I_{\ass}(\eta_2)$ \\
	$\M, \bm{t}, \ass \models \neg \varphi$ & iff & $\M, \bm{t}, \ass \not\models \varphi$ \\
	$\M, \bm{t}, \ass  \models \varphi \land \psi $ & iff & $\M, \bm{t}, \ass  \models \varphi$ and $\M, \bm{t}, \ass  \models \psi$ \\
	$\M, \bm{t}, \ass \models \exists x\varphi $ & iff & $\M, \bm{t}, \ass[x\mapsto o] \models \varphi$ for some $o\in \D$\\
	$\M, \bm{t}, \ass \models (\lambda x\psi)(\eta)$ & iff & $\M, \bm{t}, \ass[x\mapsto o] \models \psi$ and $o = \I_{\ass}(\eta)$ \\[.5ex]
	$\M, \bm{t}, \ass \models (\lambda x\psi)(\imath y\varphi)$ & iff & \begin{minipage}[t]{6.9cm}there exists $o\in \D$ such that $\M, \bm{t}, \ass[y\mapsto o] \models \varphi$ and $\M, \bm{t}, \ass[x\mapsto o] \models \psi$, and for any $o' \in \D$, if $\M, \bm{t}, \ass[y\mapsto o'] \models \varphi$, then $o' = o$\smallskip\end{minipage}\\
	$\M, \bm{t}, \ass \models \mathsf{F}\varphi $ & iff & $\M, \bm{s}, \ass \models \varphi$ for some $\bm{s}\in\T$ such that $\bm{t}\prec \bm{s}$ \\
	$\M, \bm{t}, \ass \models \mathsf{P}\varphi $ & iff & $\M, \bm{s}, \ass \models \varphi$ for some $\bm{s}\in\T$ such that $\bm{s} \prec \bm{t}$ \\
	$\M, \bm{t}, \ass \models \bm{\eta}$ & iff & $\bm{t} = \I_{\ass}(\bm{\eta})$ \\[.5ex]
	$\M, \bm{t}, \ass \models \imath\bm{x}\varphi $ & iff & \begin{minipage}[t]{6.9cm}$\M, \bm{t}, \ass[\bm{x}\mapsto\bm{t}] \models \varphi$ and for any $\bm{s} \in \T$, if $\M, \bm{s}, \ass[\bm{x}\mapsto\bm{s}] \models \varphi$, then $\bm{s} = \bm{t}$\smallskip\end{minipage}\\
	$\M, \bm{t}, \ass \models @_{\bm{\xi}}\psi $ & iff & \begin{minipage}[t]{6.9cm}there exists $\bm{s}\in \T$ such that $\M, \bm{s}, \ass \models \bm{\xi}$ and $\M, \bm{s}, \ass \models \psi$\smallskip\end{minipage}\\
	$\M, \bm{t}, \ass \models \downarrow_{\bm{x}} \varphi$ & iff & $\M, \bm{t}, \ass[\bm{x}\mapsto\bm{t}] \models \varphi$,
\end{longtable}
\noindent where $P \in \mathsf{PRED}^n$, $\eta, \eta_1,\ldots,\eta_n \in \mathsf{FVAR}\cup\mathsf{CONS}$, $\varphi, \psi \in \mathsf{FOR}$, $x,y \in \mathsf{BVAR}$, $\bm{\eta} \in \mathsf{TVAR}\cup\mathsf{NOM}$, $\bm{\xi}\in\mathsf{TFOR}$, and $\bm{x} \in \mathsf{TVAR}$. A $\FOHTL$ formula $\varphi$ is \emph{satisfiable} if there exists a tense first-order model $\M$, a time instance $\bm{t}$ in the universe of $\M$, and an assignment $\ass$ such that $\M , \bm{t}, \ass \models \varphi$; it is \emph{true} in a tense first-order model $\M$ under an assignment $\ass$, in symbols $\M , \ass \models \varphi$, if it is satisfied by $\ass$ at all time instances in the universe of $\M$; it is \emph{valid}, in symbols $\models \varphi$, if, for all tense first-order models $\M$ and assignments $\ass$, it is true in $\M$ under $\ass$; it \emph{globally entails} $\psi$ in $\FOHTL$ if, for every tense first-order model $\M$ and assignment $\ass$, if $\varphi$ is true in $\M$ under $\ass$, then $\psi$ is true in $\M$ under $\ass$; it \emph{locally entails} $\psi$ if, for every tense first-order model $\M$, time instance $\bm{t}$ in the universe of $\M$, and assignment $\ass$, if $\M , \bm{t} , \ass \models \varphi$, then $\M , \bm{t}, \ass \models \psi$. 

We can obtain different underlying temporal structures by imposing suitable restrictions on $\prec$, such as, for instance, transitivity, irreflexivity, connectedness etc.

\begin{example}\label{ex::BaldKing}
	Let us consider a simplified Russellian example of the \emph{bald king of France}, formalised as $(\lambda xB(x))(\imath yK(y))$ to see how $\FOHLLDK$ deals with several recognisable problems. Consider a model $\M=(\T,\prec,\D,\I)$, depicted in \Cref{fig::BaldKing}, with:
	\begin{itemize}
		\item $\T = \{\bm{t}_0, \bm{t}_1, \bm{t}_2, \bm{t}_3, \bm{t}_4\}$,
		\item $\bm{t}_0\prec\bm{t}_1$, $\bm{t}_0\prec\bm{t}_2$, $\bm{t}_2\prec\bm{t}_3$, $\bm{t}_3\prec\bm{t}_1$, $\bm{t}_3\prec\bm{t}_4$,
		\item $\D = \{o_1,o_2\}$,
		\item $\I(B,\bm{t}_0)=\emptyset$, $\I(K,\bm{t}_0)=\{o_1\}$, $\I(B,\bm{t}_1)=\I(K,\bm{t}_1)=\{o_1\}$, $\I(B,\bm{t}_2)=\D$, $\I(K,\bm{t}_2)=\emptyset$, $\I(B,\bm{t}_3)=\I(K,\bm{t}_3)=\D$, $\I(B,\bm{t}_4)=\I(K,\bm{t}_4)=\{o_2\}$.
	\end{itemize}
	We discard $\I$ which is unessential for our needs, but define an assignment $\ass[\circ\mapsto o_1,\bullet\mapsto\bm{t}_0]$ which maps all variables to $o_1$ and all tense-variables to $\bm{t}_0$. One may easily check that $(\lambda xB(x))(\imath yK(y))$ is satisfied at $\bm{t}_1$ and $\bm{t}_4$ but for different objects, namely for $o_1$ and $o_2$, respectively, since descriptions are non-rigid terms.
	At the remaining time instances it is false, hence $\neg (\lambda xB(x))(\imath yK(y))$ is satisfied there. Note, however, that $(\lambda x\neg B(x))(\imath yK(y))$ is satisfied at $\bm{t}_0$ since it holds of $o_1$. So there is no difference between saying that the king is not bald here or that it is not the case that he is bald. On the other hand, at $\bm{t}_2$ and $\bm{t}_3$ it is also false that $(\lambda x\neg B(x))(\imath yK(y))$ since there is no unique king there. Also $\future(\lambda x B(x))(\imath yK(y))$ is satisfied at $\bm{t}_0$ and at $\bm{t}_3$ even $\going(\lambda x B(x))(\imath yK(y))$ holds, whereas at $\bm{t}_2$ $\future(\lambda x B(x))(\imath yK(y))$ is false.
\end{example}
\begin{figure}
	\centering
	\begin{tikzpicture}
		\node[circle,fill=black,draw=black,inner sep=0,minimum size=3pt,outer sep=3pt,label=270:$\substack{\bm{t}_0\\[3pt]\{K(o_1)\}}$] (a) at (0,0) {};
		\node[circle,fill=black,draw=black,inner sep=0,minimum size=3pt,outer sep=3pt,label=90:$\substack{\{B(o_1),B(o_2)\}\\[3pt]\bm{t}_2}$] (b) at (2,2) {};
		\node[circle,fill=black,draw=black,inner sep=0,minimum size=3pt,outer sep=3pt,label=270:$\substack{\bm{t}_1\\[3pt]\{K(o_1),B(o_1)\}}$] (c) at (6,0) {};
		\node[circle,fill=black,draw=black,inner sep=0,minimum size=3pt,outer sep=3pt,label=90:$\substack{\{K(o_1),B(o_1)\\\ \ \, K(o_2),B(o_2)\}\\[3pt]\bm{t}_3}$] (d) at (4,2) {};
		\node[circle,fill=black,draw=black,inner sep=0,minimum size=3pt,outer sep=3pt,label=90:$\substack{\{K(o_2),B(o_2)\}\\[3pt]\bm{t}_4}$] (e) at (6,2) {};
		\draw[thick,->] (a) -- (b);
		\draw[thick,->] (a) -- (c);
		\draw[thick,->] (b) -- (d);
		\draw[thick,->] (d) -- (e);
		\draw[thick,->] (d) -- (c);
	\end{tikzpicture}
	\caption{Tense model from Example~\ref{ex::BaldKing}}
	\label{fig::BaldKing}
\end{figure}
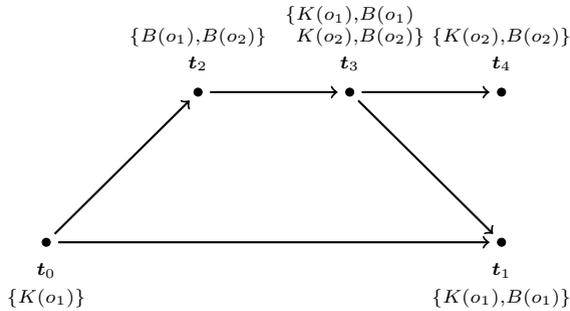

\section{Tableau Calculus}
\label{sec::Tableaux}

Several proof systems, including tableaux, sequent calculi and natural deduction, were provided for different versions of $\HL$ (see, e.g., Bra\"uner~(\citeyear{Brauner2011}), Indrzejczak~(\citeyear{Indrzejczak2010}), Zawidzki~(\citeyear{Zawidzki2014})).
Most of them represent so-called sat-calculi where each formula is preceded by the satisfaction operator.
Using sat-calculi instead of calculi working with arbitrary formulas is justified by the fact that $\varphi$ holds in (any) hybrid logic iff $@_{\bm{j}}\varphi$ holds, provided that ${\bm{j}}$ is not present in $\varphi$. And so, proving $@_{\bm{j}}\varphi$ is in essence equivalent to proving $\varphi$. In what follows we present a sat-tableau calculus for the logic $\FOHLLDK$, which we denote by $\TCHFLK$. It is in principle the calculus of Blackburn and Marx (\citeyear{BlMar03}) enriched with rules for DD and the lambda operator. Strictly speaking it is not a pure sat-calculus, since equality formulas are admitted also without satisfaction operators. Before we proceed to discussing the rules of $\TCHFLK$, let us briefly recall basic notions from the tableau methodology.

A \emph{tableau} $\Tab$ generated by a calculus $\TCHFLK$ is a \emph{derivation tree} whose nodes are assigned formulas in the language of deduction. A \emph{branch of $\Tab$} is a simple path from the root to a leaf of $\Tab$. For simplicity, we will identify each branch $\B$ with the set of formulas assigned to nodes on $\B$.

A general form of rules is as follows: $\frac{\Phi}{\Psi_1 | \ldots | \Psi_n}$, where $\Phi$ is the set of \emph{premises} and each $\Psi_i$, for $i\in\{1,\ldots,n\}$, is a set of \emph{conclusions}. If a rule has more than one set of conclusions, it is called a \emph{branching} rule. Otherwise it is \emph{non-branching}. Thus, if a rule $\frac{\Phi}{\Psi_1 | \ldots | \Psi_n}$ is applied to $\Phi$ occurring on $\B$, $\B$ splits into $n$ branches: $\B \cup \{\Psi_1\}, \ldots, \B \cup \{\Psi_n\}$. A rule $(\sf R)$ with $\Phi$ as the set of its premises is \emph{applicable} to $\Phi$ occurring on a branch $\B$ if it has not yet been applied to $\Phi$ on $\B$. A set $\Phi$ is called \emph{$(\sf R)$-expanded} on $\B$ if $(\sf R)$ has already been applied to $\Phi$ on $\B$.
A term $\xi$ is called \emph{fresh} on a branch $\B$ if it has not yet occurred on $\B$. We call a branch $\B$ \emph{closed} if the inconsistency symbol $\bot$ occurs on $\B$. If $\B$ is not closed, it is \emph{open}. A branch is \emph{fully expanded} if it is closed or no rules are applicable to (sets of) formulas occurring on $\B$. A tableau $\Tab$ is called closed if all of its branches are closed. Otherwise $\Tab$ is called open. Finally, $\Tab$ is fully expanded if all its branches are fully expanded.
A \emph{tableau proof} of a formula $\varphi$ is a closed tableau with $\neg @_{\bm{i}}\varphi$ at its root, where $\bm{i}$ is a nominal not occurring in $\varphi$. A formula $\varphi$ is tableau-valid (with respect to the calculus $\TCHFLK$) if there exists a tableau proof of~$\varphi$. The calculus $\TCHFLK$ is \emph{sound} (with respect to the semantics of $\FOHLLDK$) if, whenever a $\FOHLLDK$-formula $\varphi$ is tableau-valid, then $\varphi$ is valid. $\TCHFLK$ is \emph{complete} (with respect to the semantics of $\FOHLLDK$) if, whenever a $\FOHLLDK$-formula $\varphi$ is valid, then $\varphi$ is tableau-valid.

\subsection{Basic rules}
\label{subsec::BasicRules}
\begin{figure}
	\footnotesize
	\textbf{Closure rule}\hspace{4cm}\textbf{Propositional rules}\medskip
	
	\begin{center}
		\hspace{1mm}$(\bot)\ \dfrac{\parbox{.4cm}{\centering $ \varphi$\\$\neg\varphi$}}{\bot}$\hfill
		$(\neg)\ \dfrac{@_{\bm{j}} \neg \varphi}{\neg @_{\bm{j}} \varphi}$\qquad$(\neg\neg)\ \dfrac{\neg @_{\bm{j}} \neg\varphi}{@_{\bm{j}} \varphi}$\qquad$(\land)\ \dfrac{@_{\bm{j}}(\varphi \land \psi)}{\parbox{1cm}{\centering $@_{\bm{j}}\varphi$\\$@_{\bm{j}}\psi$}}$\qquad
		$(\neg\land)\ \dfrac{\neg @_{\bm{j}}(\varphi \land \psi)}{\neg @_{\bm{j}} \varphi \mid \neg @_{\bm{j}} \psi}$
	\end{center}\smallskip
	
	\hspace{1cm}\textbf{Quantifier rules}\hspace{5cm}\textbf{Basic modal rules}\medskip
	
	\begin{center}
		$(\exists)\ \dfrac{@_{\bm{j}} \exists x \varphi}{@_{\bm{j}}\varphi[x/a]}$\qquad
		$(\neg\exists)\ \dfrac{\neg @_{\bm{j}}\exists x \varphi}{\neg @_{\bm{j}} \varphi[x/b]}$\hfill%
		$(\future)\ \dfrac{@_{\bm{j}}\future \varphi}{\parbox{1cm}{\centering $@_{\bm{j}} \future {\bm{i}}$\\$@_{\bm{i}} \varphi$}}$\qquad
		$(\neg\future)\ \dfrac{\parbox{1.2cm}{\centering $\neg @_{\bm{j}_1} \future \varphi$\\$@_{\bm{j}_1} \future {\bm{j}_2}$}}{\neg @_{\bm{j}_2}\varphi}$\\[5pt]
		\hfill$(\past)\ \dfrac{@_{\bm{j}}\past \varphi}{\parbox{1cm}{\centering $@_{\bm{i}} \future {\bm{j}}$\\$@_{\bm{i}} \varphi$}}$\qquad
		$(\neg\past)\ \dfrac{\parbox{1.2cm}{\centering $\neg @_{\bm{j}_1} \past \varphi$\\$@_{\bm{j}_2} \future {\bm{j}_1}$}}{\neg @_{\bm{j}_2}\varphi}$
	\end{center}\smallskip
	
	{\hfill\textbf{Nominal rules}\hfill}\medskip
	
	\begin{center}
		$(\mathsf{gl})\ \dfrac{@_{\bm{j}_1} @_{\bm{j}_2} \varphi}{@_{\bm{j}_2} \varphi}$\qquad $(\neg\mathsf{gl})\ \dfrac{\neg @_{\bm{j}_1} @_{\bm{j}_2} \varphi}{\neg @_{\bm{j}_2} \varphi}$\qquad $(\mathsf{ref}_{\bm{j}})\ \dfrac{}{@_{\bm{j}} {\bm{j}}}$\qquad
		$(\mathsf{nom})\ \dfrac{\parbox{0.6cm}{\centering $@_{\bm{j}_1} {\bm{j}_2}$\\$@_{\bm{j}_1}\varphi$}}{@_{\bm{j}_2}\varphi}$\qquad
		$(\mathsf{bridge})\ \dfrac{\parbox{0.75cm}{\centering $@_{\bm{j}_1} {\bm{j}_2}$\\$@_{\bm{j}_3}\future {\bm{j}_1}$}}{@_{\bm{j}_3}\future {\bm{j}_2}}$\medskip

		$(\downarrow)\ \dfrac{@_{\bm{j}}\!\downarrow_{\bm{x}}\varphi}{@_{{\bm{j}}} \varphi[\bm{x}/\bm{j}]}$\qquad$(\neg\!\downarrow)\ \dfrac{\neg @_{\bm{j}}\!\downarrow_{\bm{x}}\varphi}{\neg @_{\bm{j}} \varphi[\bm{x}/\bm{j}]}$\qquad$(\mathsf{eq})\ \dfrac{@_{\bm{j}} b_1=b_2}{b_1=b_2}$\qquad$(\neg\mathsf{eq})\ \dfrac{\neg @_{\bm{j}} b_1=b_2}{b_1 \neq b_2}$
	\end{center}\smallskip
	
	{\hfill\textbf{$\mathbf{\imath}$-object rules}\hfill}\medskip
	
	\begin{center}
		$(\imath_1^o)\ \dfrac{@_{\bm{j}}(\lambda x \psi)(\imath y \varphi)}{\parbox[t]{1.3cm}{\centering$@_{\bm{j}}\varphi[y/a]$\\$@_{\bm{j}}\psi[x/a]$}}$ \qquad
		$(\imath_2^o)\ \dfrac{\parbox{1.8cm}{\centering $@_{\bm{j}}(\lambda x \psi)(\imath y \varphi)$\\ $@_{\bm{j}}\varphi[y/b_1]$\\ $@_{\bm{j}}\varphi[y/b_2]$}}{b_1 = b_2}$\qquad
		$(\neg \imath^o)\ \dfrac{\neg @_{\bm{j}}(\lambda x\psi)(\imath y \varphi)}{\neg @_{\bm{j}} \psi[x/b] \mid \neg @_{\bm{j}} \varphi[y/b] \mid \parbox[t]{1.15cm}{\centering$@_{\bm{j}}\varphi[y/a]$\\$a \neq b$}}$
	\end{center}\smallskip
	
	{\hfill\textbf{$\mathbf{\imath}$-temporal rules}\hfill}\medskip
	
	\begin{center}
		$(\imath_1^t)\ \dfrac{@_{\bm{j}}\imath\bm{x}\varphi}{@_{\bm{j}}\varphi[\bm{x}/\bm{j}]}$ \qquad
		$(\imath_2^t)\ \dfrac{\parbox{1.5cm}{\centering $@_{\bm{j}_1}\imath\bm{x}\varphi$\\ $@_{\bm{j}_2}\varphi[\bm{x}/\bm{j}_2]$}}{@_{\bm{j}_1}\bm{j}_2}$\qquad
		$(\neg \imath^t)\ \dfrac{\neg @_{\bm{j}}\imath\bm{x}\varphi}{\neg @_{\bm{j}}\varphi[\bm{x}/\bm{j}] \mid \parbox[t]{1.15cm}{\centering$@_{\bm{i}}\varphi[\bm{x}/\bm{i}]$\\$\neg@_{\bm{j}}\bm{i}$}}$\medskip
		
		$(@\imath^t)\ \dfrac{@_{\bm{j}}@_{\imath\bm{x}\varphi}\psi}{\parbox{1.3cm}{\centering $@_{\bm{i}}\imath\bm{x}\varphi$\\$@_{\bm{i}}\psi$}}$\qquad$(\neg@\imath^t)\ \dfrac{\neg @_{\bm{j}_1}@_{\imath\bm{x}\varphi}\psi}{\neg @_{\bm{j}_2}\imath\bm{x}\varphi\mid \neg @_{\bm{j}_2}\psi}$
	\end{center}\smallskip
	
	\hspace{2cm}\textbf{$\mathbf{\lambda}$-rules}\hfill\textbf{Other rules}\qquad\quad\phantom{a}\medskip
	
	\begin{center}
		$(\lambda)\ \dfrac{@_{\bm{j}}(\lambda x \psi)(b)}{@_{\bm{j}}\psi[x/b]}$ \qquad $(\neg\lambda)\ \dfrac{\neg @_{\bm{j}}(\lambda x \psi)(b)}{\neg @_{\bm{j}} \psi[x/b]}$\hfill$(\mathsf{ref})\ \dfrac{}{b=b}$\qquad
		$(\mathsf{RR})\ \dfrac{\parbox{1cm}{\centering $\varphi[b_1]$\\$b_1 = b_2$}}{\varphi[b_1/\!/b_2]}$\\[8pt]
		\hfill$(\mathsf{NED})\ \dfrac{}{a=a}$\qquad\quad\phantom{a}
	\end{center}
	
	\caption{Rules of the tableau calculus $\TCHFLK$}
	\label{fig::TCHFLK}
\end{figure} 
In \Cref{fig::TCHFLK} we present the rules constituting $\TCHFLK$. We transfer the notation from the previous section with the caveat that $a$ denotes an object free variable that is \emph{fresh} on the branch, whereas $b, b_1, b_2$ denote object free variables or individual constants that have already been present on the branch. Similarly, $\bm{i}$ denotes a nominal that is fresh on the branch, while, $\bm{j}, \bm{j}_1, \bm{j}_2, \bm{j}_3$ are nominals that have previously occurred on the branch. Recall that we are only considering sentences, that is, formulas without free variables (both object and tense). Consequently, even though there exist satisfaction conditions for formulas of the form $\bm{x}$ or $@_{\bm{x}}\varphi$, where $\bm{x}$ is a tense variable, the presented calculus does not comprise any rules that handle such formulas occurring independently on a branch, as such a scenario cannot materialise under the above assumption.
The closure rules and the rules handling conjunction are self-explanatory, however the remaining ones deserve a brief commentary. The rules $(\neg)$ and $(\neg\neg)$ capture self-duality of the $@$-operator. The quantifier rules are standard rules for possibilistic quantifiers ranging over the domain of a model. Bear in mind that bound variables can only be substituted with free variables or constants, but not definite descriptions, when a quantifier rule is applied. The rules $(\future)$, $(\neg\future)$, $(\past)$, and $(\neg\past)$ are standard rules for temporal modalities relying on a hybrid representation of two time instances being linked by the temporal precedence relation. More precisely, in a model $\M$ a time instance $\bm{t}_2$ represented by a nominal $\bm{j}_2$ occurs \emph{after} a time instance $\bm{t}_1$ represented by a nominal $\bm{j}_1$ if and only if a formula $@_{\bm{j}_1} \future \bm{j}_2$ holds true in $\M$. With regard to the nominal rules, $(\mathsf{gl})$ and $(\neg\mathsf{gl})$ capture a global range of $@$, that is, if a formula preceded by the $@$-operator is satisfied at one time instance in a model $\M$, it is satisfied at all time instances in $\M$. The rule $(\mathsf{ref}_{\bm{j}})$ guarantees that every nominal $\bm{j}$ is satisfied at a time instance labelled by $\bm{j}$. The \emph{bridge} rules $(\mathsf{nom})$ and $(\mathsf{bridge})$ ensure that if a nominal $\bm{j}_1$ is satisfied at a time instance labelled by $\bm{j}_2$, then $\bm{j}_1$ and $\bm{j}_2$ are interchangeable. The rules $(\downarrow)$ and $(\neg\!\downarrow)$ embody the semantics of the $\downarrow_{\bm{x}}$-operator which fixes the denotation of $\bm{x}$ to be the state the formula is currently evaluated at. More concretely, if a $\downarrow_{\bm{x}}$- (or $\neg\downarrow_{\bm{x}}$-)formula is evaluated at a time instance labelled by $\bm{j}$, that is, is preceded by $@_{\bm{j}}$, then $\bm{x}$ is substituted with $\bm{j}$ in the formula in the scope of $\downarrow_{\bm{x}}$. The rules $(\mathsf{eq})$ and $(\neg\mathsf{eq})$ reflect the fact that the object constants have the same denotations at all time instances in the model. The rule $(\imath_1^o)$ handles the scenario where, at a given time instance, an object definite description occurs in the scope of a $\lambda$-expression. Then $(\imath_1^o)$ enforces that both the formulas hold of the same fresh object constant, at the same time instance. If, moreover, a formula constituting an object definite description occurs independently on the branch, preceded with a nominal representing a given time instance, then $(\imath_2^o)$ guarantees that all the free variables or constants it holds of at this time instance denote the same object. If at a given time instance a $\lambda$-expression $\lambda x \psi$ does not hold of an object definite description $\imath y \varphi$, then for any constant $b$ present on the branch, either $\varphi$ does not hold of $b$ at this time instance or $\psi$ does not hold of $b$ at this time instance, or we can introduce a fresh constant $a$ distinct from $b$ such that $\varphi$ holds of $a$ at this time instance. The rules for temporal definite descriptions work in the following way. The rule $(\imath_1^t)$ unpacks a temporal definite description at the time instance of its evaluation. The rule $(\imath_2^t)$ guarantees that a time instance satisfying the formula which constitutes a temporal definite description is unique. According to $(\neg\imath^t)$, if a temporal definite description is not satisfied at a time instance, then either a formula constituting this description is not satisfied there or it is satisfied at a different time instance. The rule $(@\imath^t)$ reduces a formula $@_{\imath\bm{x}\varphi}\psi$ being satisfied at a time instance to the temporal definite description $\imath\bm{x}\varphi$ and a formula $\psi$ being satisfied at some (not necessarily distinct) time instance. The rule $(\neg@\imath^t)$ guarantees that if a formula $@_{\imath\bm{x}\varphi}\psi$ is not satisfied at some time instance, then $\imath\bm{x}\varphi$ and $\psi$ cannot be jointly satisfied at any time instance.
Note that these rules play an analogous role to that of $(\mathsf{gl})$ and $(\neg\mathsf{gl})$, but this time with a tense definite description in place of $\bm{j}_2$. In this case, however, the rule does not make this description the argument of the leftmost $@$-operator, as, by the construction of a proof tree, these must always be labelled by nominals.
The rules $(\lambda)$ and $(\neg\lambda)$ are tableau-counterparts of the standard $\beta$-reduction known from the $\lambda$-calculus. Their application is restricted to constants.
Finally, the $(\mathsf{ref})$ guarantees that $=$ is reflexive over all constants occurring on the branch and
$(\mathsf{RR})$ is a standard replacement rule. The rule $(\mathsf{NED})$ is a counterpart of the non-empty domain assumption. Mind, however, that it is only applied if no other rules are applicable and if no formula of the form $b=b$ is already present on the branch. Notice that also the rules $(\mathsf{ref}_{\bm{j}})$ and $(\mathsf{ref})$ do not explicitly indicate premises, however it is assumed that a nominal $\bm{j}$ or an object constant $b$ must have previously been present on the branch.

\begin{figure}
	\centering\small
	\begin{tikzpicture}
		\node[draw=none,fill=none,text width=5cm,align=center] (a) at (0,0) {$@_{\bm{j}_1}@_{\imath \bm{x}W(t,j)}M(t,j,l)$\\[3pt]
			$@_{\bm{j}_1}@_{\imath \bm{x}W(t,j)}\imath \bm{y}B$\\[3pt]				
			$\neg @_{\bm{j}_1}@_{\imath \bm{y}B}M(t,j,l)$};
		
		\node[draw=none,fill=none,text width=5cm,align=center] (b) at (0,-2) {$@_{\bm{j}_2}\imath \bm{x}W(t,j)$\\[3pt]
			$@_{\bm{j}_2}\imath \bm{y}B$};
		
		\node[draw=none,fill=none,text width=5cm,align=center] (c) at (0,-3.7) {$@_{\bm{j}_3}\imath \bm{x}W(t,j)$\\[3pt]			
			$@_{\bm{j}_3}M(t,j,l)$};
		
		\node[draw=none,fill=none] (d) at (-1,-5.4) {$\neg@_{\bm{j}_3}\imath \bm{y}B$};
		
		\node[draw=none,fill=none] (e) at (-1,-6.4) {$@_{\bm{j}_3}W(t,j)$};
		
		\node[draw=none,fill=none] (f) at (-1,-7.4) {$@_{\bm{j}_2}\bm{j}_3$};
		
		\node[draw=none,fill=none] (g) at (-1,-8.4) {$@_{\bm{j}_3}\imath \bm{y}B$};
		
		\node[draw=none,fill=none] (h) at (-1,-9.4) {$\bot$};

		\node[draw=none,fill=none] (i) at (1,-5.4) {$\neg@_{\bm{j}_3}M(t,j,l)$};
		
		\node[draw=none,fill=none] (j) at (1,-6.4) {$\bot$};
		
		\node[draw=none,fill=none] (k) at (0,-10.5) {\normalsize(a)};
		
		\draw (a) -- (b) node[midway,right] {\footnotesize$(@\imath^t)$};
		\draw (b) -- (c) node[midway,right] {\footnotesize$(@\imath^t)$};
		\draw (c) -- (d) node[pos=0.2,right,xshift=-3pt,yshift=-5pt] {\footnotesize$(\neg@\imath^t)$};
		\draw (c) -- (i);
		\draw (i) -- (j) node[midway,right] {\footnotesize$(\bot)$};
		\draw (d) -- (e) node[midway,left] {\footnotesize$(\imath_1^t)$};
		\draw (e) -- (f) node[midway,left] {\footnotesize$(\imath_2^t)$};
		\draw (f) -- (g) node[midway,left] {\footnotesize$(\mathsf{nom})$};
		\draw (g) -- (h) node[midway,left] {\footnotesize$(\bot)$};
	\end{tikzpicture}\hspace{-6em}
	\begin{tikzpicture}
		\node[draw=none,fill=none,text width=5cm,align=center] (a) at (0,0) {$@_{\bm{j}_1}@_{\imath \bm{x}\varphi}\imath \bm{y}\psi$\\[3pt]
			$@_{\bm{j}_1}@_{\imath \bm{x}\varphi}\chi$\\[3pt]
			$\neg@_{\bm{j}_1} @_{\imath \bm{y}\psi}\chi$};
		
		\node[draw=none,fill=none,text width=5cm,align=center] (b) at (0,-2) {$@_{\bm{j}_2}\imath \bm{x}\varphi$\\[3pt]
			$@_{\bm{j}_2}\imath \bm{y}\psi$};
		
		\node[draw=none,fill=none,text width=5cm,align=center] (c) at (0,-3.7) {$@_{\bm{j}_3}\imath \bm{x}\varphi$\\[3pt]
			$@_{\bm{j}_3}\chi$};
		
		\node[draw=none,fill=none] (d) at (0,-5) {$@_{\bm{j}_2}\varphi[\bm{x}/\bm{j}_2]$};
		
		\node[draw=none,fill=none] (e) at (0,-6) {$@_{\bm{j}_2}\bm{j}_3$};
		
		\node[draw=none,fill=none] (f) at (0,-7) {$@_{\bm{j}_2}\chi$};
		
		\node[draw=none,fill=none] (g) at (-1,-8.4) {$\neg @_{\bm{j}_2}\imath \bm{y}\psi$};
		\node[draw=none,fill=none] (i) at (1,-8.4) {$\neg @_{\bm{j}_2}\chi$};
		
		\node[draw=none,fill=none] (h) at (-1,-9.4) {$\bot$};
		\node[draw=none,fill=none] (j) at (1,-9.4) {$\bot$};
		
		\node[draw=none,fill=none] (k) at (0,-10.5) {\normalsize(b)};

		\draw (a) -- (b) node[midway,right] {\footnotesize$(@\imath^t)$};
		\draw (b) -- (c) node[midway,right] {\footnotesize$(@\imath^t)$};
		\draw (c) -- (d) node[midway,right] {\footnotesize$(\imath_1^t)$};
		\draw (d) -- (e) node[midway,right] {\footnotesize$(\imath_2^t)$};
		\draw (e) -- (f) node[midway,right] {\footnotesize$(\mathsf{nom})$};
		\draw (f) -- (g) node[pos=0.3,right,xshift=-5pt,yshift=-5pt] {\footnotesize$(\neg@\imath^t)$};
		\draw (f) -- (i);
		\draw (i) -- (j) node[midway,right] {\footnotesize$(\bot)$};
		\draw (g) -- (h) node[midway,left] {\footnotesize$(\bot)$};
	\end{tikzpicture}\hspace{-8em}
	\begin{tikzpicture}
		\node[draw=none,fill=none,text width=5cm,align=center] (a) at (0,0.25) {$@_{\bm{j}_1}\future(\exists x\varphi)$\\[3pt]
			$\neg@_{\bm{j}_1}\exists x\future\varphi$};
		
		\node[draw=none,fill=none,text width=5cm,align=center] (b) at (0,-1.7) {$@_{\bm{j}_1}\future \bm{j}_2$\\[3pt]
			$@_{\bm{j}_2}\exists x\varphi$};
		
		\node[draw=none,fill=none] (c) at (0,-3.2) {$@_{\bm{j}_2}\varphi[x/a]$};
		
		\node[draw=none,fill=none] (d) at (0,-4.4) {$\neg@_{\bm{j}_1}\future\varphi[x/a]$};
		
		\node[draw=none,fill=none] (e) at (0,-5.4) {$\neg@_{\bm{j}_2}\varphi[x/a]$};
		
		\node[draw=none,fill=none] (f) at (0,-6.4) {$\bot$};
		
		\node[draw=none,fill=none] (g) at (0,-10.5) {\normalsize(c)};

		\draw (a) -- (b) node[midway,right] {\footnotesize$(\future)$};
		\draw (b) -- (c) node[midway,right] {\footnotesize$(\exists)$};
		\draw (c) -- (d) node[midway,right] {\footnotesize$(\neg\exists)$};
		\draw (d) -- (e) node[midway,right] {\footnotesize$(\neg\future)$};
		\draw (e) -- (f) node[midway,right] {\footnotesize$(\bot)$};
		
	\end{tikzpicture}
	\caption{Example proofs conducted in $\TCHFLK$; \Cref{fig::Examples}(a) shows a proof tree for the derivation $@_{\imath \bm{x}W(t,j)}M(t,j,l),\ @_{\imath \bm{x}W(t,j)}\imath \bm{y}B\ \vdash\ @_{\imath \bm{y}B}M(t,j,l)$ from \Cref{ex::Wedding}; \Cref{fig::Examples}(b) presents a proof of the derivability of the rule $(\mathsf{DD})$~$@_{\bm{j}} @_{\imath \bm{x}\varphi}\imath \bm{y}\psi\ ,\ @_{\bm{j}}@_{\imath \bm{x}\varphi}\chi\ / \ @_{\bm{j}} @_{\imath \bm{y}\psi}\chi$ in $\TCHFLK$; \Cref{fig::Examples}(c) displays a proof of the validity of the Barcan formula in $\HFLK$.}
	\label{fig::Examples}
\end{figure}

\begin{example}\label{ex::Wedding}
	We provide a simple example to illustrate the application of our rules for tense definite descriptions\footnote{Examples illustrating the use of ordinary definite descriptions were provided by Indrzejczak and Zawidzki in their previous work~(\citeyear{IndZaw2023}).}. Consider the following valid argument:
	\begin{quote}
		At the year of their wedding Tricia and John moved to London. The wedding day of Tricia and John and the Brexit happened at the same year. Hence they moved to London at the year of Brexit.
	\end{quote}
	It may be formalised in a simplified form (avoiding details not relevant for the validity of this example) in the following way:
	$$@_{\imath \bm{x}W(t,j)}M(t,j,l),\ @_{\imath \bm{x}W(t,j)}\imath \bm{y}B\ \vdash\ @_{\imath \bm{y}B}M(t,j,l).$$
	As shown by the tableau proof displayed in Figure~\ref{fig::Examples}(a), the above reasoning is indeed valid in $\FOHTL$.
\end{example}

\begin{example}
	Let us consider a definite description-counterpart of the rule $(\mathsf{nom})$ defined on nominals, namely the rule:
	$$(\mathsf{DD})\  \dfrac{\parbox{1.9cm}{\centering $@_{\bm{j}} @_{\imath \bm{x}\varphi}\imath \bm{y}\psi$\\$@_{\bm{j}}@_{\imath \bm{x}\varphi}\chi$}}{@_{\bm{j}} @_{\imath \bm{y}\psi}\chi}.$$
	In Figure~\ref{fig::Examples}(b) we show, using $\TCHFLK$, that $(\mathsf{DD})$ is derivable in $\TCHFLK$.	
\end{example}

\begin{example}
	Since in $\FOHLLDK$ we assume that the object domain is common for all time instances, that is, we make the constant domain assumption, the Barcan formula should be valid in this logic. This is indeed the case, which is proved in Figure~\ref{fig::Examples}(c) with the following instance of the Barcan formula:
	$$\future(\exists x\varphi)\rightarrow\exists x\future\varphi.$$
\end{example}

\section{Soundness and Completeness}
\label{sec::Completeness}

In what follows, we will be using two auxiliary results (whose standard proofs by induction on the complexity of $\varphi$ are omitted).

\begin{lemma}[Coincidence Lemma]\label{lem::Coincidence}
	Let $\varphi \in \mathsf{FOR}$, let $\M = (\T, \prec,\D,\I)$ be a tense model, let $\bm{t} \in \T$, and let $\ass_1,\ass_2$ be assignments. If $\ass_1(\mathcal{x})=\ass_2(\mathcal{x})$ for each $\mathcal{x} \in \mathsf{FVAR} \cup \mathsf{TVAR}$ occurring in $\varphi$, then $\M, \bm{t}, \ass_1 \models \varphi$ iff 
	$\M, \bm{t}, \ass_2 \models \varphi$.
\end{lemma}

\begin{lemma}[Substitution Lemma]\label{lem::Substitution}
	Let $\varphi \in \mathsf{FOR}$, $a\in\mathsf{FVAR}$, $\bm{i} \in \mathsf{NOM}$, let $\M = (\T, \prec,\D,\I)$ be a tense model, and let $\bm{t} \in \T$. Then
	$\M, \bm{t}, \ass \models \varphi[x/a]$ iff $\M, \bm{t}, \ass[x\mapsto \ass(a)]\models \varphi$, where $x\in\mathsf{BVAR}$. Similarly, $\M, \bm{t}, \ass \models \varphi[\bm{x}/\bm{i}]$ iff $\M, \bm{t}, \ass[\bm{x}\mapsto\I(\bm{i})]\models \varphi$, where $\bm{x}\in\mathsf{TVAR}$.
\end{lemma}

\subsection{Soundness}\label{subsect::Soundndess}

Let $(\mathsf{R})$ $\frac{\Phi}{\Psi_1\mid\ldots\mid\Psi_n}$ be a rule from $\TCHFLK$. We say that $(\mathsf{R})$ is \emph{sound} if whenever $\Phi$ is satisfiable, then $\Phi\cup\Psi_i$ is satisfiable, for some $i \in \{1,\ldots,n\}$. It holds that:

\begin{lemma}\label{lem::SoundRules}
	All rules of $\TCHFLK$ are sound.
\end{lemma}

\begin{proof}
	Since closure, propositional, quantifier rules, basic modal, and most nominal rules are standard and proved to be sound elsewhere (see, e.g., Bra\"uner (\citeyear{Brauner2011})), below we only present proofs of soundness of	$\imath$-object rules, $\imath$-temporal rules and $\lambda$-rules.\smallskip
	
	\noindent $(\imath_1^o)$\quad Assume that $@_{\bm{j}}(\lambda x\psi)(\imath y \varphi)$ is satisfiable. It means that there exists a model $\M = (\T, \prec,\D,\I)$, a time instance $\bm{t} \in \T$, and an assignment $\ass$ such that $\M, \bm{t}, \ass \models @_{\bm{j}}(\lambda x \psi)(\imath y \varphi)$. Hence, by the satisfaction condition for $@$-formulas, there exists a time instance $\bm{t}' \in \T$ such that $\I(\bm{j}) = \bm{t}'$ and $\M,\bm{t}',\ass \models (\lambda x \psi)(\imath y \varphi)$. Thus, there is an object $o \in \D$ such that $\M, \bm{t}, \ass[y\mapsto o] \models \varphi$ and $\M, \bm{t}, \ass[x\mapsto o] \models \psi$, and for any $o' \in \D$, if $\M, \bm{t}, \ass[y\mapsto o'] \models \varphi$, then $o' = o$. Without loss of generality let's assume that $a$ is a fresh free variable such that $\ass(a) = o$. By the \hyperref[lem::Substitution]{Substitution Lemma}, we get that $\M, \bm{t}', \ass \models \varphi[y/a],\psi[x/a]$. Finally, it means that $\M, \bm{t}, \ass \models @_{\bm{j}}\varphi[y/a],@_{\bm{j}}\psi[x/a]$, as expected.\smallskip
	
	\noindent $(\imath_2^o)$\quad Assume that $@_{\bm{j}}(\lambda x\psi)(\imath y \varphi)$, $@_{\bm{j}}\varphi[y/b_1]$, and $@_{\bm{j}}\varphi[y/b_2]$ are joint\-ly satisfiable. It means that there exists a model $\M =(\T, \prec,\D,\I)$, a time instance $\bm{t} \in \T$, and an assignment $\ass$ such that $\M, \bm{t}, \ass \models @_{\bm{j}}(\lambda x \psi)(\imath y \varphi), @_{\bm{j}}\varphi[y/b_1], @_{\bm{j}}\varphi[y/b_2]$. Thus, by the satisfaction condition for $@$-formulas we imply that there is a time instance $\bm{t}' \in \T$ such that $\I(\bm{j}) = \bm{t}'$ and $\M , \bm{t}' , \ass \models (\lambda x\psi)(\imath y \varphi), \varphi[y/b_1], \varphi[y/b_2]$. And so, there exists an object $o \in \D$ such that $\M, \bm{t}' ,\ass[x,y\mapsto o] \models \varphi, \psi$ and, for any $o' \in \D$ if $\M, \bm{t}', \ass[y\mapsto o'] \models \varphi$, then $o'=o$. Let $\ass(b_1) = o'$ and $\ass(b_2) = o''$. By the \hyperref[lem::Substitution]{Substitution Lemma}, $\M, \bm{t}', \ass[y\mapsto o' \models \varphi$ and $\M, \bm{t}', \ass[y\mapsto o''] \models \varphi$. Since $x$ does not occur freely in $\varphi$, by the \hyperref[lem::Coincidence]{Coincidence Lemma} we get $\M , \bm{t}', \ass[x\mapsto o,y\mapsto o'] \models \varphi$ and $\M , \bm{t}' , \ass[x\mapsto o,y\mapsto o''] \models \varphi$. By the relevant satisfaction condition we obtain $o' = o$ and $o'' = o$, and so, $o = o' = o''$. As $o = \ass(b_1) = \ass(b_2)$, the respective satisfaction conditions yield $\M , \bm{t}'', \ass \models b_1 = b_2$ for any $\bm{t}'' \in \T$, so in particular, $\M, \bm{t}, \ass \models  b_1=b_2$.\smallskip
	
	\noindent $(\neg \imath^o)$\quad Assume that $\neg @_{\bm{j}}(\lambda x \psi)(\imath y \varphi)$ is satisfiable. Then there exists a model $\M = (\T, \prec,\D,\I)$, a time instance $\bm{t} \in \T$, and an assignment $\ass$ such that $\M , \bm{t}, \ass \models \neg @_{\bm{j}}(\lambda x \psi)(\imath y \varphi)$, and so, $\M , \bm{t}, \ass \not\models @_{\bm{j}}(\lambda x \psi)(\imath y \varphi)$. By the relevant satisfaction condition for $@$-formulas there exists a time instance $\bm{t}' \in \T$ such that $\I(\bm{j}) = \bm{t}'$ and $\M , \bm{t}' , \ass \not\models (\lambda x \psi)(\imath y \varphi)$. Consequently, it means that for all objects $o \in \D$ (at least) one of the following three conditions holds: \begin{enumerate*}[label=(\arabic*)]\item $\M, \bm{t}', \ass[x\mapsto o] \not\models \psi$; \item $\M, \bm{t}', \ass[y\mapsto o] \not\models \varphi$; \item there exists $o'\in\D$ such that $\M , \bm{t}', \ass[y\mapsto o'] \models \varphi$ and $o'\neq o$.\end{enumerate*} Let $b$ be a free variable present on the branch and $\ass(b) = o'$. If (1) holds for $o'$, that is, $\M, \bm{t}', v[x\mapsto o'] \not\models \psi$, then, by the \hyperref[lem::Substitution]{Substitution Lemma}, $\M, \bm{t}', \ass \not\models \psi[x/b]$, whence, by the respective satisfaction condition, we get $\M, \bm{t}', \ass \models \neg \psi[x/b]$ and, subsequently, $\M, \bm{t}, \ass \models @_{\bm{j}}\neg \psi[x/b]$ and $\M, \bm{t}, \ass \models \neg @_{\bm{j}} \psi[x/b]$. Let (2) hold for $o'$, that is, $\M , \bm{t}', v[y\mapsto o'] \not\models \varphi$. By the \hyperref[lem::Substitution]{Substitution Lemma} we get $\M , \bm{t}', \ass \not\models \varphi[y/b]$. By the satisfaction conditions for negation and $@$-formulas we obtain, subsequently, $\M , \bm{t}', \ass \models \neg\varphi[y/b]$, $\M , \bm{t}, \ass \models @_{\bm{j}}\neg\varphi[y/b]$, and $\M , \bm{t}, \ass \models \neg @_{\bm{j}}\varphi[y/b]$. Assume that (3) holds for $o'$, that is, there exists $o''\in\D$ such that $\M, \bm{t}',  \ass[y\mapsto o''] \models \varphi$ and $o'' \neq o'$. Without loss of generality we may assume that there exists $a \in \mathsf{FVAR}$ such that $a$ does not occur freely in $\varphi$ and $\ass(a) = o''$. Since $x$ does not occur freely in $\varphi$, we can apply the \hyperref[lem::Substitution]{Substitution Lemma} twice and from $\M , \bm{t}', \ass[y\mapsto o''] \models \varphi$ obtain $\M, \bm{t}', \ass \models \varphi[y/a]$ and further $\M, \bm{t}, \ass \models @_{\bm{j}}\varphi[y/a]$.\smallskip
	
	\noindent $(\imath_1^t)$\quad Assume that $@_{\bm{j}}\imath\bm{x}\varphi$ is satisfiable. It means that there exists a model $\M = (\T, \prec,\D,\I)$, a time instance $\bm{t} \in \T$, and an assignment $\ass$ such that $\M, \bm{t}, \ass \models @_{\bm{j}}\imath\bm{x}\varphi$. Hence, by the relevant satisfaction conditions for $@$-formulas, there exists a time instance $\bm{t}' \in \T$ such that $\I(\bm{j}) = \bm{t}'$ and $\M,\bm{t}',\ass \models \imath\bm{x}\varphi$, and further, $\M, \bm{t}', \ass[\bm{x}\mapsto\bm{t}'] \models \varphi$. Without loss of generality let's assume that $\bm{i}$ is a fresh nominal such that $\ass(\bm{i}) = \bm{t}'$. By the \hyperref[lem::Substitution]{Substitution Lemma}, we get that $\M, \bm{t}', \ass \models \varphi[\bm{x}/\bm{i}]$. Finally, it means that $\M, \bm{t}, \ass \models @_{\bm{j}} \varphi[\bm{x}/\bm{i}]$, as required.\smallskip
	
	\noindent $(\imath_2^t)$\quad Assume that $@_{\bm{j}_1}\imath\bm{x}\varphi$ and $@_{\bm{j}_2}\varphi[\bm{x}/\bm{j}_2]$ are joint\-ly satisfiable. It means that there exists a model $\M =(\T, \prec,\D,\I)$, a time instance $\bm{t} \in \T$, and an assignment $\ass$ such that $\M, \bm{t}, \ass \models @_{\bm{j}_1}\imath\bm{x}\varphi, @_{\bm{j}_2}\varphi[\bm{x}/\bm{j}_2]$. Thus, by the relevant satisfaction condition for $@$-formulas we imply that there are time instances $\bm{t}',\bm{t}'' \in \T$ such that $\I(\bm{j}_1) = \bm{t}'$, $\I(\bm{j}_2) = \bm{t}''$, $\M , \bm{t}' , \ass \models \imath\bm{x}\varphi$, and $\M , \bm{t}'' , \ass \models \varphi[\bm{x}/\bm{j}_2]$. Further, by the satisfaction condition for $\imath\bm{x}\varphi$, we get that $\M , \bm{t}' , \ass[\bm{x}\mapsto\bm{t}'] \models \varphi$. By the \hyperref[lem::Substitution]{Substitution Lemma} we obtain $\M , \bm{t}'' , \ass[\bm{x}\mapsto\bm{t}''] \models \varphi$, and so, again by the same satisfaction condition as before, it follows that $\bm{t}'=\bm{t}''$. Since we have that $\M, \bm{t}',\ass \models \bm{j}_1$ and $\M, \bm{t}'',\ass \models \bm{j}_2$, by the respective satisfaction conditions we get $\M, \bm{t}', \ass \models \bm{j}_2$ and subsequently, $\M, \bm{t}, \ass \models @_{\bm{j_1}}\bm{j}_2$.\smallskip
	
	\noindent $(\neg \imath^t)$\quad Assume that $\neg @_{\bm{j}}\imath\bm{x}\varphi$ is satisfiable. Then there exists a model $\M = (\T, \prec,\D,\I)$, a time instance $\bm{t} \in \T$, and an assignment $\ass$ such that $\M , \bm{t}, \ass \models \neg @_{\bm{j}}\imath\bm{x}\varphi$, and so, $\M , \bm{t}, \ass \not\models @_{\bm{j}}\imath\bm{x}\varphi$. By the relevant satisfaction condition for $@$-formulas there exists a time instance $\bm{t}' \in \T$ such that $\I(\bm{j}) = \bm{t}'$ and $\M , \bm{t}' , \ass \not\models \imath\bm{x}\varphi$. By the satisfaction condition for $\imath\bm{x}\varphi$ it means either $\M , \bm{t}' , \ass[\bm{x}\mapsto\bm{t}'] \not\models \varphi$ or there exists a time instance $\bm{t}''\in\T$ such that $\M,\bm{t}'',\ass[\bm{x}\mapsto\bm{t}'']\models\varphi$ and $\bm{t}'\neq\bm{t}''$. In the former case, by the \hyperref[lem::Substitution]{Substitution Lemma} we get $\M , \bm{t}' , \ass \not\models \varphi$ and further, by the relevant satisfaction conditions, $\M , \bm{t}' , \ass \models \neg\varphi$, $\M , \bm{t} , \ass \models @_{\bm{j}}\neg\varphi$, and finally, $\M , \bm{t} , \ass \models \neg@_{\bm{j}}\varphi$. In the latter case assume, without loss of generality, that $\bm{i}\in\mathsf{NOM}$ is such that $\I(\bm{i}) = \bm{t}''$. Since $\bm{t}'\neq\bm{t}''$, by the respective satisfaction conditions we get, subsequently, $\M,\bm{t}'\ass\not\models\bm{i}$, $\M,\bm{t}',\ass\models\neg\bm{i}$, $\M,\bm{t},\ass\models@_{\bm{j}}\neg\bm{i}$, and $\M,\bm{t},\ass\models\neg @_{\bm{j}}\bm{i}$. Moreover, by the 
	\hyperref[lem::Substitution]{Substitution Lemma} we obtain $\M,\bm{t}''\ass\models\varphi[\bm{x}/\bm{i}]$, whence, by the relevant satisfaction condition for $@$-formulas, we derive $\M,\bm{t}\ass\models@_{\bm{i}}\varphi[\bm{x}/\bm{i}]$.\smallskip
	
	\noindent $(@\imath^t)$\quad Assume that $@_{\bm{j}}@_{\imath\bm{x}\varphi}\psi$ is satisfiable. It means that there exists a model $\M = (\T, \prec,\D,\I)$, a time instance $\bm{t} \in \T$, and an assignment $\ass$ such that $\M, \bm{t}, \ass \models @_{\bm{j}}@_{\imath\bm{x}\varphi}\psi$. Hence, by the relevant satisfaction conditions for $@$-formulas, there exist time instances $\bm{t}', \bm{t}'' \in \T$ such that $\I(\bm{j}) = \bm{t}'$ and $\M,\bm{t}',\ass \models @_{\imath\bm{x}\varphi}\psi$, and further, $\M, \bm{t}'', \ass \models \imath\bm{x}\varphi,\psi$. Without loss of generality let's assume that $\bm{i}$ is a fresh nominal such that $\ass(\bm{i}) = \bm{t}''$. Then we obtain $\M,\bm{t},\ass\models @_{\bm{i}}\imath\bm{x}\varphi, @_{\bm{i}}\psi$.\smallskip
	
	\noindent $(\neg@\imath^t)$\quad Assume that $\neg@_{\bm{j}_1}@_{\imath\bm{x}\varphi}\psi$ is satisfiable. It means that there exists a model $\M = (\T, \prec,\D,\I)$, a time instance $\bm{t} \in \T$, and an assignment $\ass$ such that $\M, \bm{t}, \ass \models \neg@_{\bm{j}_1}@_{\imath\bm{x}\varphi}\psi$. By the satisfaction condition for $\neg$ we get $\M, \bm{t}, \ass \not\models \neg@_{\bm{j}_1}@_{\imath\bm{x}\varphi}\psi$. Next, by the relevant satisfaction conditions for $@$-formulas, we know that there exists a time instance $\bm{t}' \in \T$ such that $\I(\bm{j}) = \bm{t}'$ and $\M,\bm{t}',\ass \not\models @_{\imath\bm{x}\varphi}\psi$. Let $\bm{t}'' \in \T$ and $\bm{j}_2 \in \mathsf{NOM}$ be such that $\I(\bm{j}_2) = \bm{t}''$. From the satisfaction condition for $@_{\imath\bm{x}\varphi}$ we derive that
	either $\M,\bm{t}'',\ass \not\models \imath\bm{x}\varphi$ or $\M,\bm{t}'',\ass\not\models\psi$. In the former case, by applying the relevant satisfaction conditions we obtain $\M,\bm{t},\ass\not\models @_{\bm{j}_2}\imath\bm{x}\varphi$, and finally, $\M,\bm{t},\ass\models \neg@_{\bm{j}_2}\imath\bm{x}\varphi$. In the latter case, by applying the same satisfaction conditions we derive $\M,\bm{t},\ass\not\models @_{\bm{j}_2}\psi$, and finally, $\M,\bm{t},\ass\models \neg@_{\bm{j}_2}\psi$, as expected.\smallskip
	
	\noindent $(\lambda)$\quad Let $b$ be a free variable present on the branch. Assume that $@_{\bm{j}}(\lambda x \psi)(b)$ is satisfiable. Then there exists a model $\M = (\T, \prec,\D,\I)$, a time instance $\bm{t} \in \T$, and an assignment $\ass$ such that $\M, \bm{t}, \ass \models @_{\bm{j}}(\lambda x \psi)(b)$. By the relevant satisfaction condition for $@$-formulas it holds that there exists a state $\bm{t}' \in \T$ such that $\I(i) = \bm{t}'$ and $\M, \bm{t}', \ass \models (\lambda x \psi)(b)$. By the respective satisfaction condition it means that $\ass(b) = o$, for some $o \in \D$, and $\M, \bm{t}', \ass[x\mapsto o] \models \psi$. By the \hyperref[lem::Substitution]{Substitution Lemma} it holds that $\M, \bm{t}', \ass \models \psi[x/b]$, hence $\psi[x/b]$, and thus $@_{\bm{j}} \psi[x/b]$, are satisfiable.\smallskip
	
	\noindent $(\neg\lambda)$\quad Let $b$ be a parameter present on the branch. Assume that $\neg @_{\bm{j}}(\lambda x \psi)(b)$ is satisfiable. Then there exists a model $\M = (\T, \prec,\D,\I)$, a time instance $\bm{t} \in \T$, and an assignment $\ass$ such that $\M, \bm{t}, \ass \models \neg @_{\bm{j}}(\lambda x \psi)(b)$. By the relevant satisfaction condition for $@$-formulas it means that there is a time instance $\bm{t}' \in \T$ such that $\I(i) = \bm{t}'$ and $\M ,\bm{t}', \ass \not\models (\lambda x \psi)(b)$. Assume that $\ass(b) = o$ for some $o \in \D$. Then by the respective satisfaction condition $\M, \bm{t}', \ass[x\mapsto o] \not\models \varphi$. By the \hyperref[lem::Substitution]{Substitution Lemma} we get that $\M, \bm{t}', \ass \not\models \psi[x/b]$. Again, by the satisfaction condition for negation it follows that $\M, \bm{t}', \ass \models \neg\psi[x/b]$, and finally, $\M,\bm{t}, \ass \models @_{\bm{j}}\neg\psi[x/b]$.
\end{proof}

Now we are ready to prove the following theorem:

\begin{theorem}[Soundness]\label{thm::Soundness}
	The tableau calculus $\TCHFLK$ is sound.
\end{theorem}

\begin{proof} Let $\varphi$ be a $\FOHLLDK$-formula. Let $\Tab$ be a $\TCHFLK$-proof of $\varphi$. Each branch of $\Tab$ is closed. By Lemma~\ref{lem::SoundRules} all the rules of $\TCHFLK$ preserve satisfiability, and so, going from the bottom to the top of $\Tab$, we start from unsatisfiable leafs and mark sets of formulas labelling subsequent preceding nodes as unsatisfiable, eventually reaching the root, where we have $\neg @_{\bm{j}} \varphi$. Since it is unsatisfiable, too, we obtain that $\varphi$ is valid.
\end{proof}

\subsection{Completeness}
\label{subsect::Completeness}

In this section, we show that $\TCHFLK$ is complete with respect to the semantics of $\FOHLLDK$, that is, that all valid $\FOHLLDK$ formulas have $\TCHFLK$-proofs. To this end, we show the contrapositive, i.e., that if applying $\TCHFLK$ to a formula $\neg@_{\bm{j}}\varphi$ generates a tableau with an open and fully expanded branch $\B$, i.e., a tableau which is not a proof of $\varphi$, then there exists a model satisfying $\neg\varphi$ and this model can be constructed by using information stored on $\B$.


Now, for the remainder of this section assume that $\Tab$ is a $\TCHFLK$-tableau with $\neg@_{\bm{j}}\varphi$ at the root and $\B$ is an open and fully expanded branch of $\Tab$. Let $\mathsf{FVAR}_\B$, $\mathsf{CONS}_\B$, and $\mathsf{NOM}_\B$ be, respectively, the sets of all free variables, individual constants, and nominals occurring on $\B$. Below we define relations $\sim_\B \subseteq (\mathsf{FVAR}_\B\cup\mathsf{CONS}_\B)^2$ and $\approx_\B \subseteq \mathsf{NOM}_\B^2$. Let $b_1,b_2 \in \mathsf{PAR}_\B\cup\mathsf{CONS}_\B$ and $\bm{j}_1,\bm{j}_2 \in \mathsf{NOM}_\B$. Then:
\noindent\begin{align}
	\label{rel::P}	b_1 \sim_\B b_2 &\quad\text{iff}\quad b_1=b_2 \in \B\tag{$\mathsf{P}$}\\
	\label{rel::N}	\bm{j}_1 \approx_\B \bm{j}_2 &\quad\text{iff}\quad @_{\bm{j}_1} \bm{j}_2 \in \B.\tag{$\mathsf{N}$}
\end{align}
The proposition below points to a useful property of $\sim_\B$ and $\approx_\B$ which we will rely on in the further construction of a model.

\begin{proposition}\label{prop::Equivalence}
	\begin{enumerate}
		\item $\sim_\B$ is an equivalence relation on $\mathsf{PAR}_\B\cup\mathsf{CONS}_\B$.
		\item $\approx_\B$ is an equivalence relation on $\mathsf{NOM}_\B$.
	\end{enumerate}
\end{proposition}

\begin{proof}
	\emph{Reflexivity} of $\sim_\B$ follows from the expandedness of $\B$ and the presence of $(\mathsf{ref})$ in $\TCHFLK$. \emph{Relexivity} of $\approx_\B^\mathsf{NOM}$ is a consequence of the expandedness of $\B$, which results in $@_{\bm{j}} \bm{j}$ being present on $\B$ for each $\bm{j} \in \mathsf{NOM}_\B$ (thanks to $(\mathsf{ref}_{\bm{j}})$). For the \emph{symmetry} of $\sim_\B$ assume that $b_1 = b_2 \in \B$. By the expandedness of $\B$ we know that $(\mathsf{ref})$ has been applied to $b_1$ and $b_2$, yielding $b_1=b_1,b_2=b_2 \in \B$. A single application of $(\mathsf{RR})$ to $b_1=b_1,b_1=b_2 \in \B$ results in $b_2=b_1 \in \B$. To prove the \emph{symmetry} of $\approx_\B$ assume that $\bm{j}_1 \approx_\B \bm{j}_2$. Then $@_{\bm{j}_1}\bm{j}_2 \in \B$. Since $\B$ is fully expanded, $@_{\bm{j}_1} \bm{j}_1 \in \B$ by $(\mathsf{ref}_{\bm{j}})$. Then $(\mathsf{nom})$ must have been applied to $@_{\bm{j}_1}\bm{j}_2, @_{\bm{j}_1}\bm{j}_1 \in \B$, thus yielding $@_{\bm{j}_2}\bm{j}_1 \in \B$. For the \emph{transitivity} of $\sim_\B$ assume that $b_1,b_2,b_3 \in \mathsf{PAR}_\mathsf{B}\cup\mathsf{CONS}_\B$, $b_1 \sim_\B b_2$, and $b_2 \sim_\B b_3$. If $b_1$ is identical to $b_2$ or $b_2$ is identical to $b_3$, we immediately obtain $b_1\sim_\B b_3$. If $b_1, b_2, b_3$ are pairwise distinct, the identities $b_1 = b_2$ and $b_2 = b_3$ must have occurred on $\B$. A single application of $(\mathsf{RR})$ to both of them yields $b_1 = b_3$ and since $\B$ is fully expanded, $b_1 = b_3 \in \B$. Hence, $b_1 \sim_\B b_3$. For the \emph{transitivity} of $\approx_\B$ assume that $\bm{j}_1 \approx_\B \bm{j}_2$ and $\bm{j}_2 \approx_\B \bm{j}_3$. Then $@_{\bm{j}_1}\bm{j}_2, @_{\bm{j}_2}\bm{j}_3 \in \B$. By the argument used in the proof of symmetry of $\approx_\B$, we know that $@_{\bm{j}_2}\bm{j}_1 \in \B$. Applying $(\mathsf{nom})$ to $@_{\bm{j}_2}\bm{j}_1, @_{\bm{j}_2}\bm{j}_3 \in \B$ gives us $@_{\bm{j}_1}\bm{j}_3 \in \B$, and so, $\bm{j}_1 \approx_\B \bm{j}_3$.
\end{proof}

We will now show how to use the data stored on $\B$ to construct the \emph{branch structure} $\M_\B = (\T_\B, \prec_\B, \D_\B, \I_\B)$ and \emph{branch assignment} $\ass_\B$. Let $\mathsf{PRED}_\B$, $\mathsf{BVAR}_\B$, and $\mathsf{TVAR}_\B$ denote the sets of all predicate symbols occurring on $\B$, all (bound) variables occuring on $\B$, and all tense variables occurring on $\B$, respectively.  We define $\M_\B$ in the following way:
\begin{itemize}
	\item $\T_\B$ is the set of all equivalence classes of $\approx_\B$ over $\mathsf{NOM}_\B$;
	\item For any $\bm{t}_1,\bm{t}_2 \in \T_\B$, $\bm{t}_1\prec_\B\bm{t}_2$ if and only if there exist $\bm{j}_1,\bm{j}_2 \in \mathsf{NOM}_\B$ such that $\bm{j}_1 \in \bm{t}_1, \bm{j}_2\in\bm{t}_2$ and $@_{\bm{j}_1}\future \bm{j}_2 \in \B$;
	\item $\D_\B$ is the set of all equivalence classes of $\sim_\B$ over $\mathsf{FVAR}_\B\cup\mathsf{CONS}_\B$;
	\item For any $\bm{j} \in \mathsf{NOM}_\B$, $\I_\B(\bm{j}) = \bm{t}$, for $\bm{t} \in \T_\B$ such that $\bm{j} \in \bm{t}$;
	\item For any $i \in \mathsf{CONS}_\B$, $\I_\B(i) = o$, for $o \in \D_\B$ such that $i \in o$;
	\item For any $n \in \mathbb{N}^+$, $n$-ary predicate symbol $P \in \mathsf{PRED}_\B$ and $\bm{t} \in \T_\B$, $I_\B(P,\bm{t}) = \{\langle o_1,\ldots,o_n\rangle \in (\D_{\B})^n\mid @_{\bm{j}} P(b_1,\ldots,b_n) \in \B \text{ and } \bm{j} \in \bm{t} \text{ and } b_1\in o_1,\ldots,b_n\in o_n\}$;
\end{itemize}
Let $\bm{j_0}$ be an arbitrarily chosen element of $\T_\B$ and let $o_0$ be an arbitrarily chosen element of $\D_\B$. Since $\B$ is an open branch for a formula $@_{\bm{j}}\varphi$, we are guaranteed that $\T_\B$ is non-empty, and therefore, such $\bm{j}_0$ can be picked. Moreover, thanks to the rule $(\mathsf{NED})$ $\D_\B$ is also non-empty and the existence of $o_0$ is secured too. By the branch assignment $\ass_\B$ we will understand a function $\ass_\B: \mathsf{FVAR}_\B \cup \mathsf{BVAR}_\B \cup \mathsf{TVAR}_\B \longrightarrow \D_\B \cup \T_\B$ defined as follows:
\begin{itemize}
	\item For any $x \in \mathsf{BVAR}_\B$, $\ass_\B(x) = o_0$;
	\item For any $b \in \mathsf{FVAR}_\B$, $\ass_\B(b) = o \in \D_\B$ if and only if $b \in o$;
	\item For any $\bm{x} \in \mathsf{TVAR}_\B$, $\ass(\bm{x}) = \bm{j}_0$.
\end{itemize}
Note that it is not possible that, for some $b_1,b_2\in\mathsf{FVAR}_\B\cup\mathsf{CONS}_\B$, $b_1=b_2, b_2\neq b_1 \in \B$. For assume the contrary. Then, after a single application of $(\mathsf{RR})$ to the above-mentioned pair of formulas, we would obtain $b_2\neq b_2 \in \B$, which, together with $b_2=b_2 \in \B$ (thanks to $(\mathsf{ref})$ and the expandedness of $\B$) would close $\B$. Consequently, for any $o \in \D_\B$, any $b_1,b_2 \in o$ and any $\varphi$, $\varphi[b_1] \in \B$ if and only if $\varphi[b_1/\!/b_2] \in \B$. Moreover, it cannot be the case that there exist $\bm{j}_1, \bm{j}_2 \in \mathsf{NOM}_\B$ such that $\bm{j}_1 \approx_\B \bm{j}_2$ and $\neg @_{\bm{j}_1} \bm{j}_2 \in \B$. Indeed, if $\bm{j}_1 \approx_\B \bm{j}_2$, then $@_{\bm{j}_1}\bm{j}_2 \in \B$ (we use the argument from the symmetry proof of $\approx_\B$), and so, the branch would immediately close. It is also impossible that there exist $\bm{j}_1,\bm{j}_2,\bm{j}_3 \in \mathsf{NOM}_\B$ such that $\bm{j}_2 \approx_\B \bm{j}_3$, $@_{\bm{j}_1} \future \bm{j}_2$, and $\neg@_{\bm{j}_1} \future \bm{j}_3$. If it were the case, then $@_{\bm{j}_2} \bm{j}_3$ would have to be present on $\B$. Since $\B$ is fully expanded, $(\mathsf{bridge})$ would have been applied to $@_{\bm{j}_2}\bm{j}_3,@_{\bm{j}_1} \future \bm{j}_2 \in \B$ resulting in $@_{\bm{j}_1}\future \bm{j}_3 \in \B$ and closing $\B$. Finally, by the definition of $\T_\B$ and $\D_\B$, $\ass_\B$ is defined on the whole domain.  Thus, $\M_\B$ is a well-defined model and $\ass_\B$ is a well-defined assignment.

Let $\mathsf{FOR}_\B$ be the set of all formulas $\psi$ such that $@_{\bm{j}}\psi \in \B$ for some $\bm{j} \in \mathsf{NOM}_\B$. Below we make an observation that will be of use in the remainder of the section.

\begin{fact}
	Let $b_1,b_2\in\mathsf{FVAR}_\B\cup\mathsf{CONS}_\B$ be such that $b_1\sim_\B b_2$, let $\bm{j},\bm{j}_1,\bm{j}_2\in\mathsf{NOM}_\B$ be such that $\bm{j}_1 \approx_\B \bm{j}_2$, and let $\psi \in \mathsf{FOR}_\B$. Then:
	\begin{enumerate}
		\item $@_{\bm{j}}\psi \in \B$ if and only if $@_{\bm{j}}\psi[b_1/\!/b_2] \in \B$;
		\item $@_{\bm{j}_1}\psi \in \B$ if and only if $@_{\bm{j}_2}\psi \in \B$.
	\end{enumerate}
\end{fact}

The next lemma is the pillar of the completeness theorem concluding this section

\begin{lemma}
	\label{lem::BranchModel}
	Let $\M_\B$ and $\ass_\B$ be defined as above. Then for any $\psi \in \mathsf{FOR}_\B$, $\bm{j}\in\mathsf{NOM}_\B$, and $\bm{t} \in \T_\B$ such that $\bm{j} \in \bm{t}$:
	\begin{enumerate}
		\item if $@_{\bm{j}}\psi \in \B$, then $\M_\B, \bm{t}, \ass_\B \models \psi$;\\
		\item if $\neg@_{\bm{j}}\psi \in \B$, then $\M_\B, \bm{\bm{t}}, \ass_\B \not\models\psi$.
	\end{enumerate}
\end{lemma}

\begin{proof}
	We prove the lemma by induction on the complexity of $\psi$ skipping the boolean and quantifier cases which are well known. We first show that the first implication holds.\smallskip
	
	\noindent $\psi :=P(b_1,\ldots,b_n)$\quad Assume that $@_{\bm{j}}P(b_1,\ldots,b_n) \in \B$. Assume, moreover, that $\bm{t} \in \T_\B$ is such that $\bm{j} \in \bm{t}$ and $o_1, \ldots, o_n \in \D_\B$ are such that $b_1 \in o_1, \ldots, b_n \in o_n$. By the definition of $\M_\B$, $\langle o_1,\ldots,o_n\rangle \in {\I_\B}(P,\bm{t})$, and so, $\M_\B, \bm{t}, \ass_\B \models P(b_1,\ldots,b_n)$.\smallskip
	
	\noindent $\psi :=  b_1\!=\!b_2$\quad Assume that $@_{\bm{j}}b_1=b_2 \in \B$. Assume, moreover, that $\bm{t} \in \T_\B$ is such that $\bm{j} \in \bm{t}$ and $o_1, o_2\in\D_\B$ are such that $b_1 \in o_1, b_2 \in o_2$. Since $\B$ is fully expanded, $(\mathsf{eq})$ must have been applied to $@_{\bm{j}} b_1=b_2$, thus yielding $b_1 = b_2 \in \B$. By the definition of $\sim_\B$, $o_1 = o_2$, and so, by the definition of $\M_\B$ and $\ass_\B$, ${\I_\B}_{\ass_\B}(b_1) = {\I_\B}_{\ass_\B}(b_2)$. Thus, by the satisfaction condition for $=$-formulas, $\M_\B, \bm{t}, \ass_\B \models b_1 = b_2$.\smallskip
	
	\noindent $\psi :=  \bm{j}'$\quad Assume that $@_{\bm{j}}\bm{j}' \in \B$. Assume, moreover, that $\bm{t}, \bm{t}' \in \T_\B$ are such that $\bm{j} \in \bm{t}$ and $\bm{j}' \in \bm{t}'$. By the definition of $\approx_\B$, $\bm{t} = \bm{t}'$, and so, by the definition of $\M_\B$ and the satisfaction condition for nominals, $\M_\B, \bm{t}, v_\B \models \bm{j}'$.\smallskip
	
	\noindent $\psi := (\lambda x\chi)(b)$\quad Assume that $@_{\bm{j}}(\lambda x\chi)(b)\in \B$. Assume, moreover, that $\bm{t} \in \T_\B$ is such that $\bm{j} \in \bm{t}$ and $o\in\D_\B$ is such that $b \in o$. Since $\B$ is fully expanded, the rule $(\lambda)$ must have been applied to $@_{\bm{j}}\psi$, yielding $@_{\bm{j}}\chi[x/b] \in \B$. By the inductive hypothesis, $\M_\B, \bm{t}, \ass_\B \models \chi[x/b]$. By the \hyperref[lem::Substitution]{Substitution Lemma}, $\M_\B, \bm{t}, \ass_\B[x\mapsto o] \models \chi$, which, together with the fact that ${\I_\B}_{\ass_\B}(b) = o$, gives $\M_\B, \bm{t}, \ass_\B \models (\lambda x \chi)(b)$.\smallskip
	
	\noindent $\psi := (\lambda x\chi)(\imath y \theta)$\quad Assume that $@_{\bm{j}}(\lambda x\chi)(\imath y \theta)\in \B$. Assume, moreover, that $\bm{t} \in \T_\B$ is such that $\bm{j} \in \bm{t}$. Due to the expandedness of $\B$, the rule $(\imath_1^o)$ must have been applied to $@_{\bm{j}}\psi$, yielding $@_{\bm{j}}\chi[x/a], @_{\bm{j}}\theta[y/a] \in \B$. By the inductive hypothesis, $\M_\B, \bm{t}, \ass_\B \models \chi[x/a],\theta[y/a]$. Let $o \in \D_\B$ be such that $a \in o$. By the \hyperref[lem::Substitution]{Substitution Lemma}, $\M_\B, \bm{t}, \ass_\B[x\mapsto o] \models \chi, \theta[y/a]$. Now, let $b\in\mathsf{FVAR}_\B\cup\mathsf{CONS}_\B$ be such that $@_{\bm{j}}\theta[y/b] \in \B$. Then $(\imath_2)$ was applied to $@_{\bm{j}}\theta[y/a]$ and $@_{\bm{j}}\theta[y/b]$ yielding $@_{\bm{j}}a=b \in \B$. By the inductive hypothesis, ${\I_\B}_{\ass_\B}(b)={\I_\B}_{\ass_\B}(a)=o$ and $\M_\B, \bm{t}, \ass_\B \models \theta[y/b]$, and so, by the \hyperref[lem::Substitution]{Substitution Lemma}, $\M_\B, \bm{t}, \ass_\B[y\mapsto o] \models \theta$. Since $b$ is arbitrary, we get $\M_\B, \bm{t}, \ass_\B \models (\lambda x \chi)(\imath y \theta)$.\smallskip
	
	
	
	
	\noindent $\psi := \future \chi$\quad Assume that $@_{\bm{j}}\future \chi \in \B$. Assume, moreover, that $\bm{t} \in \T_\B$ is such that $\bm{j} \in \bm{t}$. By the expandedness of $\B$ it follows that $(\future)$ was applied to $@_{\bm{j}}\future \chi$ yielding $@_{\bm{j}}\future \bm{j}', @_{\bm{j}'}\chi \in \B$. Let $\bm{t}'$ be such that $\bm{j}' \in \bm{t}'$. By the inductive hypothesis we obtain $\M_\B, \bm{t}', \ass_\B \models \chi$. By the construction of $\prec_\B$ and the fact that $@_{\bm{j}}\future\bm{j}' \in \B$ we get $\bm{t}\prec_\B\bm{t}'$. Thus, by the satisfaction condition for $\future$, we arrive at $\M_\B, \bm{t}, \ass_\B \models \future \chi$.\smallskip
	
	\noindent $\psi := \past \chi$\quad We proceed similarly to the previous case.\smallskip
	
	\noindent $\psi := \imath\bm{x}\chi$\quad Assume that $@_{\bm{j}} \imath\bm{x}\chi \in \B$. Assume, moreover, that $\bm{t} \in \T_\B$ is such that $\bm{j} \in \bm{t}$. Since $\B$ is fully expanded, $(\imath_1^t)$ must have been applied to $@_{\bm{j}}\imath\bm{x}\chi$, which resulted in $@_{\bm{j}} \chi[\bm{x}/\bm{j}] \in \B$. By the inductive hypothesis, $\M_\B, \bm{t}, \ass_\B \models \chi[\bm{x}/\bm{j}]$. Now, let $\bm{j}' \in \mathsf{NOM}_\B$ be such that $@_{\bm{j}'}\chi[\bm{x}/\bm{j}'] \in \B$ and let $\bm{t}'\in\T_\D$ be such that $\bm{j}' \in \bm{t}'$. Then, by the expandedness of $\B$, $(\imath_2^t)$ was applied to $@_{\bm{j}} \imath\bm{x}\chi$ and $@_{\bm{j}'} \chi$ resulting in $@_{\bm{j}}\bm{j}' \in \B$. By the definition of $\approx_\B$ and $\T_\B$, $\bm{t} = \bm{t'}$. Since $\bm{j}'$ (and therefore, $\bm{t}'$) was arbitrary, the respective satisfaction condition is satisfied, and so, $\M_\B, \bm{t}, \ass_\B \models \imath\bm{x} \chi$.\smallskip
	
	\noindent $\psi := @_{\bm{j}'} \chi$\quad Assume that $@_{\bm{j}} @_{\bm{j}'} \chi \in \B$. Assume, moreover, that $\bm{t},\bm{t}' \in \T_\B$ are such that $\bm{j} \in \bm{t}$ and $\bm{j}' \in \bm{t}'$. Since $\B$ is fully expanded, $(\mathsf{gl})$ must have been applied to $@_{\bm{j}}@_{\bm{j}'} \chi$, which resulted in $@_{\bm{j}'} \chi \in \B$. By the inductive hypothesis, $\M_\B, \bm{t}, \ass_\B \models \chi$. By the satisfaction condition for $@_{\bm{j}}$-formulas we obtain $\M_\B, \bm{t}, \ass_\B \models @_{\bm{j}'} \chi$.\smallskip
	
	\noindent $\psi := @_{\imath\bm{x}\chi} \theta$\quad Assume that $@_{\bm{j}} @_{\imath\bm{x}\chi} \theta \in \B$. Assume, moreover, that $\bm{t} \in \T_\B$ is such that $\bm{j} \in \bm{t}$. Since $\B$ is fully expanded, $(@\imath^t)$ must have been applied to $@_{\bm{j}} @_{\imath\bm{x}\chi} \theta$, which resulted in $@_{\bm{i}} \imath\bm{x}\chi, @_{\bm{i}}\theta \in \B$. By the inductive hypothesis, $\M_\B, \bm{t}, \ass_\B \models \imath\bm{x}\chi$ and $\M_\B, \bm{t}, \ass_\B \models\theta$. By the satisfaction condition for $@_{\imath\bm{x}\varphi}$-formulas we obtain $\M_\B, \bm{t}, \ass_\B \models @_{\imath\bm{x}\chi} \theta$.\smallskip
	
	\noindent $\psi := \downarrow_{\bm{x}}\! \chi$\quad Assume that $@_{\bm{j}}\!\downarrow_{\bm{x}}\! \chi \in \B$. Assume, moreover, that $\bm{t}\in \T_\B$ is such that $\bm{j} \in \bm{t}$. Since $\B$ is fully expanded, $(\downarrow)$ was applied to $@_{\bm{j}}\!\downarrow_{\bm{x}}\! \chi \in \B$, which resulted in $@_{\bm{j}} \chi[\bm{x}/\bm{j}] \in \B$. By the inductive hypothesis, $\M_\B, \bm{t}, \ass_\B \models \chi[\bm{x}/\bm{j}]$. By the \hyperref[lem::Substitution]{Substitution Lemma} we obtain $\M_\B, \bm{t}, \ass_\B[\bm{x}\mapsto\bm{t}] \models @_{i'} \chi$. By the satisfaction condition for $@$-formulas we obtain $\M_\B, \bm{t}, \ass_\B \models \downarrow_{\bm{x}}\! \chi$.\smallskip
	
	We now proceed to a proof of the second implication.\smallskip
	
	\noindent $\psi :=P(b_1,\ldots,b_n)$\quad Assume that $\neg@_{\bm{j}}P(b_1,\ldots,b_n) \in \B$. Assume, moreover, that $\bm{t} \in \T_\B$ is such that $\bm{j} \in \bm{t}$ and $o_1, \ldots, o_n\in\D_\T$ are such that $b_1 \in o_1, \ldots, b_n \in o_n$. Since $\B$ is open, we know that $@_{\bm{j}}P(b_1,\ldots,b_n) \notin \B$. By the definition of $\M_\B$, $\langle o_1,\ldots,o_n\rangle \notin {\I_\B}((P,\bm{t})$, and so, $\M_\B, \bm{t}, \ass_\B \not\models P(b_1,\ldots,b_n)$.\smallskip
	
	\noindent $\psi :=  b_1=b_2$\quad Assume that $\neg@_{\bm{j}}b_1 = b_2 \in \B$. Assume, moreover, that $\bm{t} \in \T_\B$ is such that $\bm{j} \in \bm{t}$ and $o_1, o_2\in\D_\T$ are such that $b_1 \in o_1, b_2 \in o_2$. Since $\B$ is fully expanded, $(\neg\mathsf{eq})$ was applied to $@_{\bm{j}} b_1=b_2$, thus yielding $b_1 \neq b_2 \in \B$. By the openness of $\B$ and the definition of $\sim_\B$, $o_1 \neq o_2$, and so, by the definition of $\ass_\B$, $\ass_\B(b_1) \neq \ass_\B(b_2)$. Thus, by the satisfaction condition for $=$-formulas, $\M_\B, \bm{t}, \ass_\B \not\models b_1 = b_2$.\smallskip
	
	\noindent $\psi :=  \bm{j}'$\quad Assume that $\neg @_{\bm{j}}\bm{j}' \in \B$. Assume, moreover, that $\bm{t}, \bm{t}' \in \T_\B$ are such that $\bm{j} \in \bm{t}$ and $\bm{j}' \in \bm{t}'$. Since $\B$ is open, then $@_{\bm{j}}\bm{j}' \notin \B$, and so, by the definition of $\approx_\B$, $\bm{t} \neq \bm{t}'$. Thus, by the definition of $\ass_\B$ and the satisfaction conditions for nominals and $\neg$-formulas, we get $\M_\B, \bm{t}, \ass_\B \not\models \bm{j}'$.\smallskip
	
	\noindent $\psi := (\lambda x\chi)(b)$\quad Assume that $\neg @_{\bm{j}}(\lambda x\chi)(b)\in \B$. Assume, moreover, that $\bm{t} \in \T_\B$ is such that $\bm{j} \in \bm{t}$ and $o\in\D_\T$ is such that $b \in o$. Since $\B$ is fully expanded, the rule $(\neg\lambda)$ must have been applied to $@_{\bm{j}}\psi$, yielding $\neg @_{\bm{j}}\chi[x/b] \in \B$. By the inductive hypothesis, $\M_\B, \bm{t}, \ass_\B \not\models \chi[x/b]$. By the \hyperref[lem::Substitution]{Substitution Lemma}, $\M_\B, \bm{t}, \ass_\B[x\mapsto o] \not\models \chi$, which, together with the fact that $\ass_\B(b) = o$, gives $\M_\B, \bm{t}, \ass_\B \not\models (\lambda x \chi)(b)$.\smallskip
	
	\noindent $\psi := (\lambda x\chi)(\imath y \theta)$\quad Assume that $\neg @_{\bm{j}}(\lambda x\chi)(\imath y \theta)\in \B$. Moreover assume that $\bm{t} \in \T_\B$ is such that $\bm{j} \in \bm{t}$. Since $\B$ is fully expanded, the rule $(\neg\imath)$ was applied to $\neg @_{\bm{j}}\psi$, making, for any free variable $b$ present on the branch, one of the following three hold:
	\begin{enumerate*}[label=(\arabic*)] \item $\neg@_{\bm{j}}\chi[x/b] \in \B$, \item $\neg@_{\bm{j}}\theta[y/b] \in \B$, \item there is a fresh free variable $a$ such that $@_{\bm{j}}\theta[y/a], a\neq b \in \B$.
	\end{enumerate*} Let $o\in\D_\T$ be such that $b \in o$. Assume (1) is the case. By the inductive hypothesis we get $\M_\B, \bm{t}, \ass_\B \not\models \neg\chi[x/b]$. By the \hyperref[lem::Substitution]{Substitution Lemma} we obtain $\M_\B, \bm{t}, \ass_\B[x\mapsto o] \not\models \chi$.
	If (2) holds, then by, the inductive hypothesis, $\M_\B, \bm{t} , \ass_\B \not\models \theta[y/b]$. By the \hyperref[lem::Substitution]{Substitution Lemma} we obtain $\M_\B, \bm{t} , \ass_\B[y\mapsto o] \models \neg \theta$, and so, $\M_\B, \bm{t}, \ass_\B[y\mapsto o] \not\models \theta$.
	Finally, let (3) hold. Then, by the inductive hypothesis, $\M_\B, \bm{t}, \ass_\B \models \theta[y/a], a\neq b$. Let $o'\in\D_\T$ be such that $a \in o'$. By the openness of $\B$ and the definition of $\sim_\B$ we have $o'\neq o$. Since $x$ does not occur freely in $\theta[y/a]$, it holds that $\M_\B , \bm{t}, \ass_\B[x\mapsto o] \models \theta[y/a]$ By the \hyperref[lem::Substitution]{Substitution Lemma} we obtain $\M_\B , \bm{t}, \ass_\B[x\mapsto o,y\mapsto o'] \models \theta$. As previously noted, $o'\neq o$, which means, by the respective satisfaction condition, that taking these three possibilities together, we obtain $\M_\B , \bm{t}, \ass_\B[x\mapsto o,y\mapsto o'] \not\models (\lambda x\chi)(\imath y \theta)$. Neither $x$ nor $y$ occurs freely in $(\lambda x\chi)(\imath y \theta)$, so after applying the \hyperref[lem::Substitution]{Substitution Lemma} twice we obtain $\M_\B , \bm{t}, \ass_\B \not\models (\lambda x\chi)(\imath y \theta)$.\smallskip

	\noindent $\psi := \future \chi$\quad Assume that $\neg @_{\bm{j}}\future \chi \in \B$. Assume, moreover, that $\bm{t}, \bm{t}' \in \T_\B$ are such that $\bm{j} \in \bm{t}$ and $\bm{t}\prec_\B\bm{t}'$. Assume that $\bm{j}' \in \mathsf{NOM}_\B$ is such that $\bm{j}'\in\bm{t}'$. By the definition of $\prec_\B$, it must be the case that $@_{\bm{j}}\future\bm{j}' \in \B$. Since $\B$ is fully expanded,  $(\neg\future)$ was applied to $\neg@_{\bm{j}}\future \chi$ and $@_{\bm{j}}\future \bm{j}'$ yielding $\neg @_{\bm{j}'}\chi \in \B$. By the inductive hypothesis we obtain $\M_\B, \bm{t}, v_\B \not\models \chi$. By the construction of $\prec_\B$ and the satisfaction condition for $\future$, we get $\M_\B, \bm{t}, \ass_\B \not\models \future \chi$.\smallskip
	
	\noindent $\psi := \past \chi$\quad We proceed similarly to the previous case.\smallskip
	
	\noindent $\psi := \imath\bm{x}\chi$\quad Assume that $\neg@_{\bm{j}} \imath\bm{x}\chi \in \B$. Assume, moreover, that $\bm{t} \in \T_\B$ is such that $\bm{j} \in \bm{t}$. Since $\B$ is fully expanded, $(\neg\imath^t)$ must have been applied to $\neg @_{\bm{j}}\imath\bm{x}\chi$, which resulted either in $\neg @_{\bm{j}} \chi[\bm{x}/\bm{j}] \in \B$ or in $@_{\bm{i}}\chi[\bm{x}/\bm{i}], \neg@_{\bm{j}}\bm{i} \in \B$, for some fresh $\bm{i}\in\mathsf{NOM}_\B$. In the former case, by the inductive hypothesis we obtain $\M_\B, \bm{t}, \ass_\B \not\models \chi$. For the latter case, let $\bm{t}'\in\T_\B$ be such that $\bm{i}\in\bm{t}'$. Then, by applying the inductive hypothesis, we get $\M_\B, \bm{t}', \ass_\B \models \chi[\bm{x}/\bm{i}]$, with $\bm{t}\neq\bm{t}'$, which we know from the definition of $\T_\B$. In both cases, by the satisfaction condition for $\imath\bm{x}\varphi$, we get that $\M, \bm{t}, \ass_\B \not\models \imath\bm{x}\chi$, as required.\smallskip
	
	\noindent $\psi := @_{\bm{j}'} \chi$\quad Assume that $\neg@_{\bm{j}}@_{\bm{j}'} \chi \in \B$. Assume, moreover, that $\bm{t},\bm{t}' \in \T_\B$ are such that $\bm{j} \in \bm{t}$ and $\bm{j}' \in \bm{t}'$. Since $\B$ is fully expanded, $(\neg\mathsf{gl})$ must have been applied to $\neg@_{\bm{j}}@_{\bm{j}'} \chi$, which resulted in $\neg @_{\bm{j}'} \chi \in \B$. By the inductive hypothesis we get $\M_\B, \bm{t}', \ass_\B \not\models \chi$. By the satisfaction condition for $@$-formulas we obtain $\M_\B, \bm{t}, v_\B \not\models @_{\bm{j}'} \chi$.\smallskip
	
	\noindent $\psi := @_{\imath\bm{x}\chi} \theta$\quad Assume that $\neg@_{\bm{j}} @_{\imath\bm{x}\chi} \theta \in \B$. Assume, moreover, that $\bm{t} \in \T_\B$ is such that $\bm{j} \in \bm{t}$. Let $\bm{j}' \in \mathsf{NOM}_\B$ and $\bm{t}'\in\T_\B$ be such that $\bm{j}'\in\bm{t}'$. Since $\B$ is fully expanded, $(\neg@\imath^t)$ must have been applied to $\neg@_{\bm{j}} @_{\imath\bm{x}\chi} \theta$ and $\bm{j}'$, yielding either $\neg@_{\bm{j}'} \imath\bm{x}\chi\in\B$ or $\neg@_{\bm{j}'}\theta \in \B$. In the former case, by the inductive hypothesis, $\M_\B, \bm{t}', \ass_\B \not\models \imath\bm{x}\chi$, and in the second case, after applying the inductive hypothesis, we obtain $\M_\B, \bm{t}', \ass_\B \not\models\theta$. Since $\bm{t}'$ was an arbitrary element of $\T_\B$, by the satisfaction condition for $@_{\imath\bm{x}\varphi}$-formulas we obtain $\M_\B, \bm{t}, \ass_\B \not\models @_{\imath\bm{x}\chi} \theta$.\smallskip
	
	\noindent $\psi := \downarrow_{\bm{x}}\! \chi$\quad Assume that $\neg @_{\bm{j}}\!\downarrow_{\bm{x}}\! \chi \in \B$. Assume, moreover, that $\bm{t}\in \T_\B$ is such that $\bm{j} \in \bm{t}$. Since $\B$ is fully expanded, $(\neg \downarrow)$ was applied to $\neg @_{\bm{j}}\!\downarrow_{\bm{x}}\! \chi$, which yielded $\neg @_{\bm{j}} \chi[\bm{x}/\bm{j}] \in \B$. By the inductive hypothesis, $\M_\B, \bm{t}, \ass_\B \not\models \chi[\bm{x}/\bm{j}]$. By the satisfaction condition for $\downarrow$ we get $\M_\B, \bm{t}, \ass_\B \not\models \downarrow_{\bm{x}}\! \chi$.
\end{proof}

Lemma~\ref{lem::BranchModel} implies:

\begin{theorem}[Completeness]\label{thm::Completeness}
	The calculus $\TCHFLK$ is complete.
\end{theorem}

\begin{proof}
	Let $\varphi$ be a $\FOHLLDK$-formula. We prove the contrapositive of the completeness condition given in \Cref{sec::Tableaux}. Assume that $\varphi$ has no $\TCHFLK$-proof, that is, an application of $\TCHFLK$ to $\neg @_{\bm{j}}\varphi$ results in an open tableau. Let $\B$ be an open and fully expanded branch of such a tableau. By Lemma~\ref{lem::BranchModel} we get that $\M_\B, \bm{t}, \ass_\B \not\models \varphi$. Since $\M_\B$ and $\ass_\B$ are well defined, by the satisfaction condition for $\neg$ we get that $\M_\B, \bm{t}, \ass_\B \models \neg\varphi$
	Therefore, $\neg\varphi$ is $\FOHLLDK$-satisfiable, hence $\varphi$ is not $\FOHLLDK$-valid.
\end{proof}\medskip

This result can be extended to all logics complete with respect to any elementary class of frames $\mathscr{C}$ closed under point-generated subframes (see Blackburn and Marx~(\citeyear{BlMar03})). Every class of frames satisfying such a condition is definable by a set of nominal-free pure hybrid sentences. It follows that for every such sentence $\varphi$ it is enough to add a zero-premise rule $\frac{}{@_{\bm{j}}\varphi}$, where $\bm{j}$ is a nominal present on the branch, to retain the calculus' completeness with respect to the class of frames under consideration.

\section{Interpolation}
\label{sec::Interpolation}

In this section, we will show that $\FOHTL$ has the Craig interpolation property, that is, for any $\FOHTL$-formulas $\varphi$ and $\psi$ such that $\models\varphi\to\psi$, there exists a $\FOHTL$-formula $\chi$ such that all predicates and constants occurring in $\chi$ occur in both $\varphi$ and $\psi$, and moreover $\models\varphi\to\chi$ and $\models\chi\to\psi$. Similarly to Blackburn and Marx~(\citeyear{BlMar03}) we exploit a technique introduced by Smullyan~(\citeyear{Smullyan1968}) and further adjusted to the tableaux setting by Fitting~(\citeyear{Fitting1996}). This allows us to refer to many details of their work but to make the proof comprehensible we must recall how this strategy works. Let us consider a closed tableau for a valid implication $\varphi\to\psi$ in $\TCHFLK$. It can be mechanically transformed into a \emph{biased} tableau in the following way. We delete the root: $\neg @_{\bm{j}}(\varphi\to\psi)$, replace $@_{\bm{j}}\varphi$ with  $\texttt{L}\ @_{\bm{j}}\varphi$ and $\neg @_{\bm{j}}\psi$ with $\texttt{R}\ \neg @_{\bm{j}}\psi$, and continue the process of assigning prefixes $\texttt{L}, \texttt{R}$: for each application of a rule we precede with $\texttt{L}$ all conclusions of the premise prefixed with $\texttt{L}$ and with $\texttt{R}$ all conclusions of the $\texttt{R}$-premise. This way all formulas, save $\bot$ at the end of each branch, are signed in a way that makes explicit their ancestry: they follow either from the antecedent $\texttt{L}\ @_{\bm{j}}\varphi$ or from the succedent $\texttt{R}\ \neg @_{\bm{j}}\psi$ of the original implication. Thus in the case of rules with one premise we must always consider two variants: the $\texttt{L}$-variant and the $\texttt{R}$-variant. In the case of rules with two premises the situation is slightly more complicated since we must additionally consider the variants that have premises with opposite signs. This is the way the proof is carried out by Blackburn and Marx. To save space we refer to their work when calculating interpolants for all cases except the new ones. 
However, in contrast to their solution, in the case of the rules for definite descriptions we modify their technique in a way which guarantees that we always have to make only two calculations for each applied rule. This is reasonable since in the case of $(\imath_2$) there are three premises and, accordingly, eight variants of the rule for computing the interpolant are needed, which complicates things considerably. Instead, we can replace each rule with multiple premises with a rule having only one premise, which enables us to consider only two variants. This can also be done for two-premise rules from Blackburn and Marx's calculus, but we confine ourselves to changing only the new multi-premise rules: $(\imath^o_2)$ and $(\imath^t_2)$. Consider the following transformed rules $({\imath^o_2}')$ and $({\imath^t_2}')$:

\begin{center}
	$( {\imath^o_2}') \ \dfrac{@_{\bm{j}}(\lambda x\psi)(\imath y \varphi)}{\neg@_{\bm{j}} \varphi[y/b_1] \mid \neg @_{\bm{j}}\varphi[y/b_2] \mid b_1 = b_2}$
\end{center}\medskip

\begin{center}
	$({\imath^t_2}') \ \dfrac{@_{\bm{j}_1}\imath \bm{x} \varphi}{\neg @_{\bm{j}_2}\varphi[\bm{x}/\bm{j}_2] \mid @_{\bm{j}_1}{\bm{j}_2}}$
\end{center}\medskip

Let $\TCHFLK'$ be the calculus $\TCHFLK$ with $(\imath^o_2)$ and $(\imath^t_2)$ replaced with $({\imath^o_2}')$ and $({\imath^t_2}')$. 
We need to show that $\TCHFLK$ and $\TCHFLK'$ are \emph{equivalent}, that is, that, given a set of premises $\Phi$, the sets of formulas derivable from $\Phi$ using $\TCHFLK$ and $\TCHFLK'$ are identical. To that end we will exploit the cut rule:
\begin{center}
	$(\mathsf{cut})\ \dfrac{\phantom{\top}}{\varphi \mid \neg\varphi}$
\end{center}
Recall that a rule $(\mathsf{R})$ is \emph{admissible} for a calculus $\mathscr{C}$ if the set of theorems provable in $\mathscr{C}\cup\{(\mathsf{R})\}$ is the same as the set of theorems provable in $\mathscr{C}$. Then the following holds:
\begin{proposition}\label{prop::cut}
	$(\mathsf{cut})$ is admissible in $\TCHFLK$.
\end{proposition}
\noindent It is a straightforward consequence of $\TCHFLK$'s completeness and the fact that $(\mathsf{cut})$ is a sound rule.  Thus, we can apply $(\mathsf{cut})$ safely in $\TCHFLK$ to show the derivability of other rules and obtain:

\begin{lemma}\label{lem::equivalence}
	$\TCHFLK$ and $\TCHFLK'$ are equivalent.
\end{lemma}

\begin{proof}
	It is enough to prove that $({\imath^o_2}')$ and $({\imath^t_2}')$ are derivable in $\TCHFLK$ (with admissible $(\mathsf{cut})$), whereas $(\imath^o_2)$ and $(\imath^t_2)$ are derivable in $\TCHFLK'$. For the interderivability of $({\imath^o_2})$ and $({\imath^o_2}')$ suitable derivations are displayed in \Cref{fig::Derivability}(a) and \Cref{fig::Derivability}(b), respectively.
	\begin{figure}[th]
		\centering
		\small
		\begin{tikzpicture}
			\node[fill=none,draw=none] (a) at (-1,0) {$@_{\bm{j}}(\lambda x\psi)(\imath y \varphi)$};
			
			\node[fill=none,draw=none] (b) at (-2,-1) {$@_{\bm{j}}\varphi[y/b_1]$};
			
			\node[fill=none,draw=none] (b') at (0,-1) {$\neg@_{\bm{j}}\varphi[y/b_1]$};
			
			\node[fill=none,draw=none] (c) at (-3,-2) {$@_{\bm{j}}\varphi[y/b_2]$};
			
			\node[fill=none,draw=none] (c1) at (-1,-2) {$\neg@_{\bm{j}}\varphi[y/b_2] $};
			
			\node[fill=none,draw=none] (d1) at (-3,-3) {$b_1 = b_2$};
			
			\node[fill=none,draw=none] (l) at (-1,-3.9) {\normalsize(a)};
			
			\draw[thick,->] (a) -- (b) node[midway,right,xshift=2.5pt] {\footnotesize$(\mathsf{cut})$};
			\draw[thick,->] (a) -- (b');
			\draw[thick,->] (b) -- (c) node[midway,right,xshift=2.5pt] {\footnotesize$(\mathsf{cut})$};
			\draw[thick,->] (b) -- (c1);
			\draw[thick,->] (c) -- (d1) node[midway,left] {\footnotesize$(\imath^o_2)$};
		\end{tikzpicture}
		\begin{tikzpicture}
			\node[fill=none,draw=none,text width=.5\linewidth,align=center] (a) at (0,0) {$@_{\bm{j}}\varphi[y/b_1]$\\
				$@_{\bm{j}}\varphi[y/b_2]$\\
				$@_{\bm{j}}(\lambda x\psi)(\imath y \varphi)$};

			\node[fill=none,draw=none] (b1) at (-2,-1.6) {$\neg@_{\bm{j}}\varphi[y/b_1]$};

			\node[fill=none,draw=none] (b2) at (0,-1.6) {$\neg@_{\bm{j}}\varphi[y/b_2]$};
			
			\node[fill=none,draw=none] (b3) at (2,-1.6) {$b_1=b_2$};
			
			\node[fill=none,draw=none] (c1) at (-2,-2.6) {$\bot$};
			
			\node[fill=none,draw=none] (c2) at (0,-2.6) {$\bot$};
			
			\node[fill=none,draw=none] (l) at (0,-3.5) {\normalsize(b)};
			
			\draw[thick,->] (a) -- (b1);
			\draw[thick,->] (a) -- (b2) node[pos=0.2,fill=white] {\footnotesize$({\imath^o_2}')$};
			\draw[thick,->] (a) -- (b3);
			\draw[thick,->] (b1) -- (c1) node[midway,fill=white,left] {\footnotesize$(\bot)$};
			\draw[thick,->] (b2) -- (c2) node[midway,fill=white,left] {\footnotesize$(\bot)$};
		\end{tikzpicture}
		\caption{The interderivability of the rules $(\iota_2^o)$ and $({\iota_2^o}')$; \Cref{fig::Derivability}(a) displays a proof of the derivability of $({\iota_2^o}')$ in $\TCHFLK$  with $(\mathsf{cut})$; \Cref{fig::Derivability}(b) shows a proof of the derivability of $({\iota_2^o})$ in $\TCHFLK'$}
		\label{fig::Derivability}
	\end{figure}

	The derivability of $(\imath^t_2)$ and $({\imath^t_2}')$ in $\TCHFLK'$ and $\TCHFLK$, respectively, follows from similar (but simpler) derivations, so we skip them.
	Hence the two calculi are equivalent. Moreover, both of them are cut-free, that is, they do not comprise any cut rules, and analytic, i.e., every formula $\psi$ such that $@_{\bm{i}}\psi$ occurs in a tableau with $@_{\bm{j}}\varphi$ at the root is a subformula of $\varphi$ or the negation of such a subformula, modulo variable replacement.
\end{proof}

\begin{theorem}[Craig interpolation]\label{thm::Interpolation}
	$\TCHFLK'$ enjoys the Craig interpolation property.
\end{theorem}

\begin{proof}
	Assume that $\varphi\to\psi$ is $\FOHLLDK$-valid. Clearly, we exclude the cases where $\varphi\equiv\bot$ or $\psi\equiv\top$ since in these cases an interpolant is trivially $\bot$ or $\top$, respectively.  For the remaining cases we build up an interpolant constructively, starting from each occurrence of $\bot$ at the end of a branch, and going up the tree. In general, at each stage we consider the last applied rule and having already established interpolants for conclusions of the applied rule, we extract an interpolant for the premise(s) with respect to all the formulas which are above on the branch. Thus, a general scheme for one-premise rules is:
	\begin{quotation}
		\noindent If $\chi_1, ..., \chi_k$ are interpolants for $\Gamma\cup \{\Psi_1\}, \ldots, \Gamma\cup\{\Psi_k\}$, then $\texttt{I}(\chi_1, \ldots, \chi_k)$ is an interpolant for $\Gamma\cup\{\varphi\}$, where $\varphi$ is the premise of the applied rule and $\Psi_1, \ldots, \Psi_k$ are all the (sets of) conclusions, and $\Gamma$ is the set of all formulas on the branch above the premise.
	\end{quotation}
	Clearly, specific principles for calculating interpolants are in two versions for each rule: the \texttt{L}-variant with $\texttt{L}$, or the \texttt{R}-variant with $\texttt{R}$ assigned to the premise and conclusions. In the case of two-premise rules considered by Blackburn and Marx there are four combinations and the scheme is more involved~(\citeyear{BlMar03}). We present the interpolants only for the new rules. Let $\varphi$ be a formula to which a rule has been applied on the branch, let $\{\Psi_1,\ldots,\Psi_k\}$ be the set of (sets of) conclusions of $\varphi$ after applying the rule, and let $\Gamma$ be the set of all formulas occurring on the branch above $\varphi$. Moreover, let $\{\gamma_1, ..., \gamma_n\}$ be the set of all formulas such that $\texttt{L}\ \gamma_i\in \Gamma$, and let $\{\delta_1, ..., \delta_m\}$ be the set of all formulas such that $\texttt{R}\ \delta_i\in \Gamma$. For each rule we are showing that for its \texttt{L}-variant:
	\begin{quotation}
		\noindent If, for every $i\leq k$, $\models \bigwedge\Psi_i\wedge\gamma_1\wedge \ldots\wedge\gamma_n\rightarrow\chi_i$ and
		$\models\chi_i\rightarrow\neg\delta_1\vee \ldots\vee\neg\delta_m$, then
		$\models\varphi\wedge\gamma_1\wedge \ldots\wedge\gamma_n\rightarrow \texttt{I}(\chi_1, \ldots, \chi_k)$;
	\end{quotation}
	and for the \texttt{R}-variant:
	\begin{quotation}
		\noindent If, for every $i\leq k$, $\models \gamma_1\wedge \ldots\wedge\gamma_n\rightarrow\chi_i$ and
		$\models\chi_i\to\neg\delta_1\vee \ldots\vee\neg\delta_m\vee\bigvee\neg\Psi_i$, then
		$\models \texttt{I}(\chi_1, \ldots, \chi_k)\to\neg\delta_1\vee \ldots\vee\neg\delta_m\vee\neg\varphi$.
	\end{quotation}
	Note that although $\varphi, \chi_j, \gamma_l$ are sat-formulas, the interpolants do not need to; we only expect them to satisfy the conditions for being interpolants.
	
	Below we state the principles for calculating interpolants for the specific rules of $\TCHFLK'$. For the remaining rules similar principles have been proposed by Fitting~(\citeyear{Fitting1996}) and Blackburn and Marx~(\citeyear{BlMar03}).\smallskip
	\begin{description}
		\item[\normalfont$(\texttt{X} \lambda)$] If $\chi$ is an interpolant for $\Gamma\cup\{\texttt{X}\ @_{\bm{j}}\psi[x/b]\}$, then $\chi$ is an interpolant for $\Gamma\cup\{\texttt{X}\ (@_{\bm{j}}\lambda x\psi)(b)\}$, for $\texttt{X} \in \{\texttt{L},\texttt{R}\}$.
		
		
		\item[\normalfont$(\texttt{X} \neg\lambda)$] If $\chi$ is an interpolant for $\Gamma\cup\{\texttt{X}\ \neg @_{\bm{j}}\psi[x/b]\}$, then $\chi$ is an interpolant for $\Gamma\cup\{\texttt{X}\ \neg @_{\bm{j}}(\lambda x\psi)(b)\}$, for $\texttt{X} \in \{\texttt{L},\texttt{R}\}$.
		
		
		\item[\normalfont$(\texttt{X} \imath_1^o)$] If $\chi$ is an interpolant for $\Gamma\cup\{\texttt{X}\ @_{\bm{j}}\psi[x/a], \texttt{X}\ @_{\bm{j}}\varphi[x/a]\}$, then $\chi$ is an interpolant for $\Gamma\cup\{\texttt{X}\ @_{\bm{j}}(\lambda x\psi)(\imath y\varphi)\}$, for $\texttt{X} \in \{\texttt{L},\texttt{R}\}$.
		
		
		\item[\normalfont$(\texttt{L} {\imath_2^o}')$] If $\chi_1$ is an interpolant for $\Gamma\cup\{\texttt{L}\ \neg@_{\bm{j}}\varphi[y/b_1]\}$, $\chi_2$ is an interpolant for $\Gamma\cup\{\texttt{L}\ \neg@_{\bm{j}}\varphi[y/b_2]\}$ and
		$\chi_3$ is an interpolant for $\Gamma\cup\{\texttt{L}\ b_1 = b_2\}$, then $\forall x\forall y(\chi_1\vee\chi_2\vee\chi_3)[b_1/x,b_2/y]$ is an interpolant for $\Gamma\cup\{\texttt{L}\ @_{\bm{j}}(\lambda x\psi)(\imath y\varphi)\}$.
		
		\item[\normalfont$(\texttt{R} {\imath_2^o}')$] If $\chi_1$ is an interpolant for $\Gamma\cup\{\texttt{R}\ \neg@_{\bm{j}}\varphi[y/b_1]\}$, $\chi_2$ is an interpolant for $\Gamma\cup\{\texttt{R}\ \neg@_{\bm{j}}\varphi[y/b_2]\}$ and
		$\chi_3$ is an interpolant for $\Gamma\cup\{\texttt{R}\ b_1 = b_2\}$, then $\exists x\exists y(\chi_1\wedge\chi_2\wedge\chi_3)[b_1/x,b_2/y]$ is an interpolant for $\Gamma\cup\{\texttt{R}\ @_{\bm{j}}(\lambda x\psi)(\imath y\varphi)\}$.
		
		\item[\normalfont$(\texttt{L} \neg\imath^o)$] If $\chi_1$ is an interpolant for $\Gamma\cup\{\texttt{L}\ \neg@_i\psi[y/b]\}$, $\chi_2$ is an interpolant for $\Gamma\cup\{\texttt{L}\ \neg@_i\varphi[y/b]\}$ and
		$\chi_3$ is an interpolant for $\Gamma\cup\{\texttt{L}\ @_i\varphi[y/a], \texttt{L}\ a\neq b\}$, then $\forall x(\chi_1\vee\chi_2\vee\chi_3)[b/x]$ is an interpolant for $\Gamma\cup\{\texttt{L}\ \neg@_i(\lambda x\psi)(\imath y\varphi)\}$.
		
		\item[\normalfont$(\texttt{R} \neg\imath^o)$] If $\chi_1$ is an interpolant for $\Gamma\cup\{\texttt{R}\ \neg@_i\psi[y/b]\}$, $\chi_2$ is an interpolant for $\Gamma\cup\{\texttt{R}\ \neg@_i\varphi[y/b]\}$ and $\chi_3$ is an interpolant for $\Gamma\cup\{\texttt{R}\ @_i\varphi[y/a], \texttt{R}\  a\neq b\}$, then $\exists x(\chi_1\wedge\chi_2\wedge\chi_3)[b/x]$ is an interpolant for $\Gamma\cup\{\texttt{R}\ \neg@_i(\lambda x\psi)(\imath y\varphi)\}$.

		\item[\normalfont$(\texttt{X} \imath^t_1)$] If $\chi$ is an interpolant for $\Gamma\cup\{ \texttt{X}\ @_{\bm{j}}\varphi[\bm{x}/\bm{j}]\}$, then $\chi$ is an interpolant for $\Gamma\cup\{\texttt{X}\ @_{\bm{j}}\imath \bm{x}\varphi)\}$, for $\texttt{X} \in \{\texttt{L},\texttt{R}\}$.
		
		\item[\normalfont$(\texttt{L} {\imath^t_2}')$] If $\chi_1$ is an interpolant for $\Gamma\cup\{\texttt{L}\ \neg@_{\bm{j}_2}\varphi[\bm{x}/\bm{j}_2]\}$ and $\chi_2$ is an interpolant for $\Gamma\cup\{\texttt{L}\ @_{\bm{j}_1}\bm{j}_2\}$, then $\downarrow_{\bm{x}}(\chi_1\vee\chi_2)[\bm{j}_2/\bm{x}]$ is an interpolant for $\Gamma\cup\{\texttt{L}\ @_{\bm{j}_1}\imath \bm{x}\varphi\}$.

		\item[\normalfont$(\texttt{R} {\imath^t_2}')$] If $\chi_1$ is an interpolant for $\Gamma\cup\{\texttt{R}\ \neg@_{\bm{j}_2}\varphi[\bm{x}/\bm{j}_2]\}$ and $\chi_2$ is an interpolant for $\Gamma\cup\{\texttt{R}\ @_{\bm{j}_1}\bm{j}_2\}$, then $\downarrow_{\bm{x}}(\chi_1\wedge\chi_2)[\bm{j}_2/\bm{x}]$ is an interpolant for $\Gamma\cup\{\texttt{R}\ @_{\bm{j}_1}\imath \bm{x}\varphi\}$.
		
		\item[\normalfont$(\texttt{L} \neg\imath^t)$] If $\chi_1$ is an interpolant for $\Gamma\cup\{\texttt{L}\ \neg@_{\bm{j}}\varphi[\bm{x}/\bm{j}]\}$ and $\chi_2$ is an interpolant for $\Gamma\cup\{\texttt{L} @_{\bm{i}}\varphi[\bm{x}/\bm{i}], \texttt{L}\neg @_{\bm{j}}\bm{i}\}$, then $\chi_1\vee\chi_2$ is an interpolant for $\Gamma\cup\{\texttt{L}\neg @_{\bm{j}}\imath \bm{x}\varphi\}$.
		
		\item[\normalfont$(\texttt{R} \neg\imath^t)$] If $\chi_1$ is an interpolant for $\Gamma\cup\{\texttt{R}\ \neg@_{\bm{j}}\varphi[\bm{x}/\bm{j}]\}$ and $\chi_2$ is an interpolant for $\Gamma\cup\{\texttt{R} @_{\bm{i}}\varphi[\bm{x}/\bm{i}], \texttt{R}\neg @_{\bm{j}}\bm{i}\}$, then $\chi_1\wedge\chi_2$ is an interpolant for $\Gamma\cup\{\texttt{R}\neg @_{\bm{j}}\imath \bm{x}\varphi\}$.

		\item[\normalfont$(\texttt{X} {@\imath}^t)$] If $\chi$ is an interpolant for $\Gamma\cup\{\texttt{X}\ @_{\bm{i}}\imath \bm{x}\varphi, \texttt{X}\ @_{\bm{i}}\psi\}$, then $\chi$ is an interpolant for $\Gamma\cup\{\texttt{X}\ @_{\bm{j}}@_{\imath \bm{x}\varphi}\psi\}$, for $\texttt{X} \in \{\texttt{L},\texttt{R}\}$.
		
		\item[\normalfont$(\texttt{L} \neg {@\imath}^t)$] If $\chi_1$ is an interpolant for $\Gamma\cup\{\texttt{L}\ \neg @_{\bm{j}_2}\imath \bm{x}\varphi\}$ and $\chi_2$ is an interpolant for $\Gamma\cup\{\texttt{L}\neg @_{\bm{j}_2}\psi\}$, then $\downarrow_{\bm{x}}(\chi_1\vee\chi_2)[\bm{j}_2/\bm{x}]$ is an interpolant for $\Gamma\cup\{\texttt{L}\ \neg @_{\bm{j}_1}@_{\imath \bm{x}\varphi}\psi\}$.
		
		\item[\normalfont$(\texttt{R} \neg {@\imath}^t)$] If $\chi_1$ is an interpolant for $\Gamma\cup\{\texttt{R}\ \neg @_{\bm{j}_2}\imath \bm{x}\varphi\}$ and $\chi_2$ is an interpolant for $\Gamma\cup\{\texttt{R}\neg @_{\bm{j}_2}\psi\}$, then $\downarrow x(\chi_1\wedge\chi_2)[\bm{j}_2/\bm{x}]$ is an interpolant for $\Gamma\cup\{\texttt{R}\ \neg @_{\bm{j}_1}@_{\imath \bm{x}\varphi}\psi\}$.
	\end{description}
	
	\noindent 
	Let us skip the straightforward cases of $(\texttt{X} \lambda)$, $(\texttt{X} \neg\lambda)$, $(\texttt{X} \imath_1^o)$ and check two harder cases: $(\texttt{L}{\imath_2^o}')$ and $(\texttt{R} {\imath^o_2}')$. 
	For the first rule we assume that:
	\begin{enumerate}
		\item \label{cond::1}$\models \neg@_{\bm{j}}\varphi[y/b_1]\wedge\gamma_1\wedge \ldots\wedge\gamma_n\rightarrow\chi_1$
		\item \label{cond::2}$\models\chi_1\rightarrow\neg\delta_1\vee ...\vee\neg\delta_m$
		\item \label{cond::3}$\models \neg@_{\bm{j}}\varphi[y/b_2]\wedge\gamma_1\wedge \ldots\wedge\gamma_n\rightarrow\chi_2$
		\item \label{cond::4}$\models\chi_2\rightarrow\neg\delta_1\vee \ldots\vee\neg\delta_m$
		\item \label{cond::5}$\models b_1 = b_2\wedge\gamma_1\wedge \ldots\wedge\gamma_n\rightarrow\chi_3$
		\item \label{cond::6}$\models\chi_3\rightarrow\neg\delta_1\vee \ldots\vee\neg\delta_m$
	\end{enumerate}
	On this basis we can show that $\forall x\forall y(\chi_1\vee \chi_2\vee\chi_3)[b_1/x, b_2/y]$ is the required interpolant. 
	Let us assume that $b_1$ or $b_2$ occur in any of $\chi_1, \chi_2, \chi_3$, but not in $\gamma_1, \ldots, \gamma_n$. It follows that they must also occur in $\neg\delta_1\vee \ldots\vee\neg\delta_m$. If the assumption fails, then the universal quantification of $\chi_1\vee \chi_2\vee\chi_3$ is either unnecessary (if $b_1$ or $b_2$ also occur in $\gamma_1, \ldots, \gamma_n$) or void (if they are not in $\chi_1, \chi_2, \chi_3$), and hence, also unnecessary. 
	From \eqref{cond::2}, \eqref{cond::4}, \eqref{cond::6} it obviously follows that $\models\forall x \forall y(\chi_1\vee \chi_2\vee\chi_3)[b_1/x, b_2/y]\to\neg\delta_1\vee \ldots\vee\neg\delta_m$. We show that from \eqref{cond::1}, \eqref{cond::3}, \eqref{cond::5} it follows that  $\models @_{\bm{j}}(\lambda x\psi)(\imath y\varphi)\wedge\gamma_1\wedge \ldots\wedge\gamma_n\rightarrow \forall x \forall y(\chi_1\vee \chi_2\vee\chi_3)[b_1/x, b_2/y]$. Assume, towards a contradiction, that this implication is not valid, that is, that in some model we have $\M, \bm{t}, \ass \models @_{\bm{j}}(\lambda x\psi)(\imath y\varphi)$, $\M, \bm{t}, \ass \models \gamma_1\wedge \ldots\wedge\gamma_n$ but
	$\M, \bm{t}, \ass \not\models \forall x\forall y(\chi_1\vee \chi_2\vee\chi_3)[b_1/x, b_2/y]$.
	Hence, 	$\M, \bm{t}, \ass[x\mapsto o_1,y\mapsto o_2] \not\models \chi_1\vee \chi_2\vee\chi_3[b_1/x, b_2/y]$. Since $b_1, b_2$ do not occur in $\gamma_1, \ldots, \gamma_n, \varphi$, by the \hyperref[lem::Coincidence]{Coincidence Lemma} $\ass[x\mapsto o_1,y\mapsto o_2]$ can be identified with $\ass$ itself where $o_1=\ass(x), o_2=\ass(y)$. Hence, by the \hyperref[lem::Substitution]{Substitution Lemma}, $\M, \bm{t}, \ass \not\models \chi_1\vee \chi_2\vee\chi_3$. 
	Since each of \eqref{cond::1}, \eqref{cond::3}, \eqref{cond::5} is valid and $\M, \bm{t}, \ass \models \gamma_1\wedge \ldots\wedge\gamma_n$, we know that $\M, \bm{t}, \ass \models @_{\bm{j}}\varphi[y/b_1]$, $\M, \bm{t}, \ass \models @_{\bm{j}}\varphi[y/b_2]$, $\M, \bm{t}, \ass \not\models b_1=b_2$. But $\M, \bm{t}, \ass \models @_{\bm{j}}(\lambda x\psi)(\imath y\varphi)$ and all formulas are satisfied at the same state named by $\bm{j}$, which leads to a contradiction.
	
	For  the \texttt{R}-variant $(\texttt{R} {\imath^o_2}')$
	we assume that:
	\begin{enumerate}[label=\arabic*'., ref=\arabic*']
		\item\label{cond::1'} $\models \gamma_1\wedge \ldots\wedge\gamma_n\to\chi_1$; 
		\item\label{cond::2'} $\models\chi_1\to\neg\delta_1\vee ...\vee\neg\delta_m\vee @_{\bm{j}}\varphi[y/b_1]$;
		\item\label{cond::3'} $\models \gamma_1\wedge \ldots\wedge\gamma_n\rightarrow\chi_2$; 
		\item\label{cond::4'} $\models\chi_2\rightarrow\neg\delta_1\vee ...\vee\neg\delta_m\vee @_{\bm{j}}\varphi[y/b_2]$;
		\item\label{cond::5'} $\models\gamma_1\wedge \ldots\wedge\gamma_n\rightarrow\chi_3$; 
		\item\label{cond::6'} $\models\chi_3\rightarrow\neg\delta_1\vee \ldots\vee\neg\delta_m\vee b_1\neq b_2$;
	\end{enumerate}
	We can show that $\exists x\exists y(\chi_1\wedge \chi_2\wedge\chi_3[b_1/x, b_2/y])$ is the required interpolant. 
	Again, let us assume that one or both of $b_1, b_2$ occur in any of $\chi_1, \chi_2, \chi_3$, but not in $\delta_1, \ldots, \delta_m$. It follows that they must also occur in $\gamma_1\wedge \ldots\wedge\gamma_n$. If, on the other hand, our assumption is false, then the existential quantification of $\chi_1\wedge \chi_2\wedge\chi_3$ is either unnecessary (if one or both of $b_1, b_2$ occur also in $\delta_1, \ldots, \delta_m$) or void (if they are not in $\chi_1, \chi_2, \chi_3$), and hence, also unnecessary. 
	From \eqref{cond::1'}, \eqref{cond::3'}, \eqref{cond::5'} it straightforwardly follows that  $\models\gamma_1\wedge \ldots\wedge\gamma_n\rightarrow \exists x \exists y(\chi_1\wedge \chi_2\wedge\chi_3[b_1/x, b_2/y])$. We show that
	\eqref{cond::2'}, \eqref{cond::4'}, \eqref{cond::6'} imply $\models\exists x \exists y(\chi_1\wedge \chi_2\wedge\chi_3[b_1/x, b_2/y])\rightarrow\neg\delta_1\vee\ldots\vee\neg\delta_m\vee\neg @_{\bm{j}}(\lambda x\psi)(\imath y\varphi)$. 
	For the sake of contradiction assume that the implication is not valid, hence in some model we have $\M, \bm{t}, \ass \models \exists x \exists y(\chi_1\wedge \chi_2\wedge\chi_3[b_1/x, b_2/y])$ but $\M, \bm{t}, \ass \not\models \neg\delta_1\vee\ldots\vee\neg\delta_m\vee\neg @_{\bm{j}}(\lambda x\psi)(\imath y\varphi)$. Thus, $\M, \bm{t}, \ass[x\mapsto o_1,y\mapsto o_2] \models \chi_1\wedge \chi_2\wedge\chi_3[b_1/x, b_2/y]$.
	Since $b_1, b_2$ do not occur in $\delta_1, \ldots, \delta_m, \varphi$, again, without loss of generality, we can assume that $o_1=\ass(x), o_2=\ass(y)$ and by the \hyperref[lem::Coincidence]{Coincidence Lemma} and \hyperref[lem::Substitution]{Substitution Lemma} we obtain $\M, \bm{t}, \ass \models \chi_1\wedge \chi_2\wedge\chi_3$. Since each of \eqref{cond::2'}, \eqref{cond::4'}, \eqref{cond::6'} is valid, but $\M, \bm{t}, \ass \not\models \neg\delta_1\vee\ldots\vee\neg\delta_m$, then $\M, \bm{t}, \ass \models @_{\bm{j}}\varphi[y/b_1]$, $\M, \bm{t}, \ass \models @_{\bm{j}}\varphi[y/b_2]$, and $\M, \bm{t}, \ass \models b_1\neq b_2$. Since $\M, \bm{t}, \ass \not\models \neg @_{\bm{j}}(\lambda x\psi)(\imath y\varphi)$ and all formulas are satisfied at the same state named by $\bm{j}$, we obtain a contradiction. One may notice that the proof was symmetric to the previous one and based on straightforward dualities.
	
	The cases of $(\texttt{L} \neg\imath^o)$ and $(\texttt{R} \neg\imath^o)$ are proven in an analogous, but simpler, way since only singular quantification is needed (if any). We also skip the straightforward case of $(\texttt{X} \imath^t_1)$ and demonstrate one of the cases for $({\imath^t_2}')$, namely $(\texttt{R} {\imath^t_2}')$ (the case of $(\texttt{L} {\imath^t_2}')$ can be handled by a symmetric argument based on obvious dualities, similarly to the preceding two proofs).
	
	We start by assuming that:
	\begin{enumerate}[label=\arabic*''., ref=\arabic*'']
		\item \label{cond::1''}$\models \gamma_1\wedge \ldots\wedge\gamma_n\rightarrow\chi_1$
		\item \label{cond::2''}$\models\chi_1\rightarrow\neg\delta_1\vee ...\vee\neg\delta_m\vee @_{\bm{j}_2}\varphi[\bm{x}/\bm{j}_2]$
		\item \label{cond::3''}$\models \gamma_1\wedge \ldots\wedge\gamma_n\rightarrow\chi_2$
		\item \label{cond::4''}$\models\chi_2\rightarrow\neg\delta_1\vee \ldots\vee\neg\delta_m\vee\neg @_{\bm{j}_1}\bm{j}_2$		
	\end{enumerate}
	On this basis we can show that $\downarrow_{\bm{x}}(\chi_1\wedge \chi_2)[\bm{j}_2/\bm{x}]$ is the required interpolant. 
	Let us assume that $\bm{j}_2$ occurs in any of $\chi_1, \chi_2$, but not in $\delta_1, \ldots, \delta_n$. It follows that it must also occur in $\gamma_1\wedge \ldots\wedge\gamma_n$. If the assumption is false, then the binding by the $\downarrow$ operator of $\chi_1\wedge\chi_2$ is either unnecessary (if $\bm{j}_2$ also occurs in 
	$\delta_1, \ldots, \delta_n$) or void (if $\bm{j}_2$ is not in $\chi_1, \chi_2$), in which case it is also unnecessary. 
	From \eqref{cond::1''}, \eqref{cond::3''} it obviously follows that $\models\gamma_1\wedge \ldots\wedge\gamma_n\rightarrow\downarrow_{\bm{x}}(\chi_1\wedge \chi_2)[\bm{j}_2/\bm{x}]$. We show that from \eqref{cond::2''} and \eqref{cond::4''} it follows that  $\models\downarrow_{\bm{x}}(\chi_1\wedge \chi_2)[\bm{j}_2/\bm{x}]\rightarrow\neg\delta_1\vee\ldots\vee\neg\delta_m\vee \neg @_{\bm{j}_1}\imath \bm{x}\varphi$. Assume, indirectly, that this is not the case, that is, in some model we have $\M, \bm{t}, \ass \models\downarrow_{\bm{x}}(\chi_1\wedge \chi_2)[\bm{j}_2/\bm{x}]$ but 
	$\M, \bm{t}, \ass \not\models \neg\delta_1\vee \ldots\vee\neg\delta_m$ and $\M, \bm{t}, \ass \models@_{\bm{j}_1}\imath \bm{x}\varphi$. By the latter and the soundness of $({\imath^t_2}')$, either $\M, \bm{t}, \ass \not\models @_{\bm{j}_2}\varphi[\bm{x}/\bm{j}_2]$ or $\M, \bm{t}, \ass \models@_{\bm{j}_1}\bm{j}_2$. From the first disjunct, (\ref{cond::2''}), and the fact that $\M, \bm{t}, \ass \not\models \neg\delta_1\vee ...\vee\neg\delta_m$,
	we have that $\M, \bm{t}, \ass \not\models \chi_1$, which implies that 
	$\M, \bm{t}, \ass \not\models\downarrow_{\bm{x}}(\chi_1\wedge \chi_2)[\bm{j}_2/\bm{x}]$. The same conclusion follows if we take the second disjunct and \eqref{cond::4''}; in both cases we run into a contradiction.
	
	We omit the proofs of the remaining cases, since they are either identical (the cases of $(\texttt{L} \neg {@\imath}^t)$ and $(\texttt{R} \neg {@\imath}^t)$) to the ones conducted above, or simpler (the cases of $(\texttt{L} \neg {\imath}^t)$ and $(\texttt{R} \neg {\imath}^t)$), or straightforward (the case of $(\texttt{X} {@\imath}^t)$).
\end{proof}

Although this result was proven only for $\FOHLLDK$, it is trivially extendable to all stronger logics mentioned at the end of~\Cref{sec::Completeness} (see Blackburn and Marx~(\citeyear{BlMar03}) for details).

\section{Conclusion}
\label{sec::Conclusion}

As we mentioned above, our tableau system is different than the sequent calculus by Indrzejczak (\citeyear{Indrzejczak2020a}), which formalises the approach of Fitting and Mendelsohn (FMA). The differences concern, \textit{inter alia}, the background theory of definite descriptions and the language of both systems.  In the case of FMA the background theory of definite descriptions is based on the axiom of Hintikka which has the following form~\citep{Indrzejczak2020a}:
\begin{align}
	\mathcal{t} \approx\imath x\varphi \leftrightarrow (\lambda x\varphi)(\mathcal{t}) \wedge \forall y(\varphi[x/y]\rightarrow (\lambda x x=y)(\mathcal{t})), \text{ for } y \text{ not in }\varphi,\tag{$\mathsf{H}$}\label{ax::Hintikka}
\end{align}
where $\approx$ represents intensional equality (as opposed to extensional $=$).

In the tableau system of Fitting and Mendelsohn, instead of rules, suitable instances of implications building (\ref{ax::Hintikka}) are simply added to branches in proof trees. Indrzejczak (\citeyear{Indrzejczak2020a}) devised special rules to handle that but due to the form of (\ref{ax::Hintikka}) they always introduce definite descriptions as arguments of intensional equality. Moreover, some additional rules for $\approx$ are needed which introduce a certain kind of restricted cut to the system. In the present calculus the rules for definite descriptions are based on the principle of Russell as formalised with the help of the $\lambda$-operator:
\begin{align}
	(\lambda x\psi)(\imath y\varphi) \leftrightarrow \exists x(\forall y(\varphi \leftrightarrow y =x )\land\psi),\tag{$\mathsf{R}_\lambda$}\label{ax::Russell}
\end{align}
where $\varphi$ does not contain free occurrences of $x$.

Russell's modified principle (\ref{ax::Russell}) leads to simpler, more general and more natural rules characterising definite descriptions. There is also an important semantic difference between FMA and the present approach. In the former definite descriptions are semantically treated as terms and characterised by means of an interpretation function:
$$\I_\ass(\imath x\varphi,\bm{t}) = o\quad \text{iff}\quad \begin{minipage}[t]{7cm}$\M,  \bm{t}, \ass[x\mapsto o] \models \varphi$ and, for any $o'\in\D$, if $\M,  \bm{t}, \ass[x\mapsto o'] \models \varphi$, then $o'=o$.\end{minipage}$$

Here we decided to make a semantic characterisation of definite descriptions an inherent part of the characterisation of $\lambda$-atoms as expressed by satisfaction clauses. Such a solution was for the first time applied in the formalisation of the Russellian approach presented by Indrzejczak and Zawidzki~(\citeyear{IndZaw2023}). It leads to simpler metalogical proofs and better reflects Russell's eliminativist spirit.

There are also several differences concerning the language of both versions of $\FOHLLD$. FMA allows us to express a difference between non-existent and non-denoting terms by means of predicates of a special kind. It also makes expressible the difference between non-rigid terms and their rigidified versions. In the context of $\FOHL$ it is naturally captured by an application of $@_{\bm{j}}$ to terms, i.e., $@_{\bm{j}} \mathcal{t}$ denotes the object that is the denotation of $\mathcal{t}$ at the state named by ${\bm{j}}$. Since our major goal was to extend Marx and Blackburn's interpolation result to the case of definite descriptions (and $\lambda$-predicates), we maximally restricted the language and used their tableau calculus (\citeyear{BlMar03}) as the basis. On the other hand, in the present system $\downarrow$ plays a central role, and is absent from FMA as well as from its $\HL$ formalisation~\citep{Indrzejczak2020a}. Of course both systems may be extended to obtain a similar effect. In particular, note that rigidification of intensional terms in the present version of $\FOHLLD$ does not need an introduction of $@_{\bm{j}}$ applied to terms. For example, to say that the present king of France is bald, $@_{\bm{j}} (\lambda x B(x))(\imath yK(y))$ is not our only option. It can be expressed with the help of the present machinery as the formula $\downarrow_{\bm{x}}\!@_{\bm{x}}(\lambda x B(x))(\imath y K(y))$. Extending both approaches to the same vocabulary and comparing their deductive behaviour requires additional work involving implementation and preparation of benchmarks.

An interesting problem stems also from the application of two variants of interpolation proofs which were mixed here. In the recent work of Indrzejczak and Zawidzki~(\citeyear{IndZaw2023}) we applied a strategy based on a preliminary transformation of all many-premise rules into their one-premise equivalents. It seems that this solution leads to a shorter proof and a more uniform calculation of interpolants. To save space, we carried out the present proof on top of the ready-result of Blackburn and Marx, based on an alternative solution relying on multiple-premise rules directly. But this leads to a multiplication of cases. It would be interesting to check how our former strategy of computing interpolants works in the context of the present $\FOHLLD$. Another interesting task is to extend the result to systems for logics above $\sf K$, but characterised by suitable rules rather than by nominal-free pure hybrid axioms, as above. A possible departure point for obtaining such a non-trivial extension could be logics complete wrt the classes of frames defined by so-called geometric formulas. Bra\"uner (\citeyear{Brauner2011}) showed how to characterise such theories in a natural deduction system for $\HL$ by rules of a uniform character. His method can be adapted to the sequent or tableau setting. The catch is that such rules contain nominals, and so calculating interpolants requires binding nominals, which makes it a less trivial task than for the case of nominal-free axioms. Eventually, it would be interesting to prove the interpolation theorem for Indrzejczak's former system~(\citeyear{Indrzejczak2020a}) which is based on significantly different rules.

\paragraph{Acknowledgements} This research is funded by the European Union (ERC, ExtenDD, project number: 101054714). Views and opinions expressed are however those of the authors only and do not necessarily reflect those of the European Union or the European Research Council. Neither the European Union nor the granting authority can be held responsible for them.

\paragraph{Conflict of interest} The author(s) has(ve) no competing interests to declare that are relevant to the content of this article.




\begin{thebibliography}{}
	\providecommand{\doi}[1]{\url{https://doi.org/#1}}
	\bibcommenthead
	
	\bibitem [\protect \citeauthoryear {%
		Areces%
		, Blackburn%
		\BCBL {}\ \BBA {} Marx%
	}{%
		Areces%
		\ \protect \BOthers {.}}{%
		{\protect \APACyear {2001}}%
	}]{%
		ArBlMar01}
	\APACinsertmetastar {%
		ArBlMar01}%
	\begin{APACrefauthors}%
		Areces, C.%
		, Blackburn, P.%
		\BCBL {} Marx, M.%
	\end{APACrefauthors}%
	\unskip\
	\newblock
	\APACrefYearMonthDay{2001}{}{}.
	\newblock
	{\BBOQ}\APACrefatitle {Hybrid Logics: Characterization, Interpolation and
		Complexity} {Hybrid logics: Characterization, interpolation and
		complexity}.{\BBCQ}
	\newblock
	\APACjournalVolNumPages{Journal of Symbolic Logic}{66}{3}{977--1010}.
	\newblock
	
	\newblock
	
	\newblock
	\begin{APACrefDOI} \doi{10.2307/2695090} \end{APACrefDOI}
	\PrintBackRefs{\CurrentBib}
	
	\bibitem [\protect \citeauthoryear {%
		Areces%
		, Blackburn%
		\BCBL {}\ \BBA {} Marx%
	}{%
		Areces%
		\ \protect \BOthers {.}}{%
		{\protect \APACyear {2003}}%
	}]{%
		ArBlMar03}
	\APACinsertmetastar {%
		ArBlMar03}%
	\begin{APACrefauthors}%
		Areces, C.%
		, Blackburn, P.%
		\BCBL {} Marx, M.%
	\end{APACrefauthors}%
	\unskip\
	\newblock
	\APACrefYearMonthDay{2003}{}{}.
	\newblock
	{\BBOQ}\APACrefatitle {Repairing the interpolation theorem in quantified modal
		logic} {Repairing the interpolation theorem in quantified modal
		logic}.{\BBCQ}
	\newblock
	\APACjournalVolNumPages{Annals of Pure and Applied Logic}{124}{1-3}{287--299}.
	\newblock
	
	\newblock
	
	\newblock
	\begin{APACrefDOI} \doi{10.1016/S0168-0072(03)00059-9} \end{APACrefDOI}
	\PrintBackRefs{\CurrentBib}
	
	\bibitem [\protect \citeauthoryear {%
		Artale%
		, Mazzullo%
		, Ozaki%
		\BCBL {}\ \BBA {} Wolter%
	}{%
		Artale%
		\ \protect \BOthers {.}}{%
		{\protect \APACyear {2021}}%
	}]{%
		ArtaleEtAl2021}
	\APACinsertmetastar {%
		ArtaleEtAl2021}%
	\begin{APACrefauthors}%
		Artale, A.%
		, Mazzullo, A.%
		, Ozaki, A.%
		\BCBL {} Wolter, F.%
	\end{APACrefauthors}%
	\unskip\
	\newblock
	\APACrefYearMonthDay{2021}{}{}.
	\newblock
	{\BBOQ}\APACrefatitle {On Free Description Logics with Definite Descriptions}
	{On free description logics with definite descriptions}.{\BBCQ}
	\newblock
	M.~Bienvenu, G.~Lakemeyer\BCBL {}\ \BBA {} E.~Erdem\ (\BEDS), \APACrefbtitle
	{Proceedings of the 18th {I}nternational {C}onference on {P}rinciples of
		{K}nowledge {R}epresentation and {R}easoning, {KR} 2021, Online event,
		{N}ovember 3--12, 2021} {Proceedings of the 18th {I}nternational {C}onference
		on {P}rinciples of {K}nowledge {R}epresentation and {R}easoning, {KR} 2021,
		online event, {N}ovember 3--12, 2021}\ (\BPGS\ 63--73).
	\newblock
	\begin{APACrefDOI} \doi{10.24963/kr.2021/7} \end{APACrefDOI}
	\PrintBackRefs{\CurrentBib}
	
	\bibitem [\protect \citeauthoryear {%
		Blackburn%
		\ \BBA {} J{\o}rgensen%
	}{%
		Blackburn%
		\ \BBA {} J{\o}rgensen%
	}{%
		{\protect \APACyear {2012}}%
	}]{%
		BlJo2012}
	\APACinsertmetastar {%
		BlJo2012}%
	\begin{APACrefauthors}%
		Blackburn, P.%
		\BCBT {}\ \BBA {} J{\o}rgensen, K.F.%
	\end{APACrefauthors}%
	\unskip\
	\newblock
	\APACrefYearMonthDay{2012}{}{}.
	\newblock
	{\BBOQ}\APACrefatitle {Indexical Hybrid Tense Logic} {Indexical hybrid tense
		logic}.{\BBCQ}
	\newblock
	T.~Bolander, T.~Bra{\"{u}}ner, S.~Ghilardi\BCBL {}\ \BBA {} L.S.~Moss\
	(\BEDS), \APACrefbtitle {Advances in {M}odal {L}ogic 9, papers from the ninth
		conference on ``{A}dvances in {M}odal {L}ogic,'' held in {C}openhagen,
		{D}enmark, 22-25 {A}ugust 2012} {Advances in {M}odal {L}ogic 9, papers from
		the ninth conference on ``{A}dvances in {M}odal {L}ogic,'' held in
		{C}openhagen, {D}enmark, 22-25 {A}ugust 2012}\ (\BPGS\ 144--160).
	\newblock
	\APACaddressPublisher{}{College Publications}.
	\PrintBackRefs{\CurrentBib}
	
	\bibitem [\protect \citeauthoryear {%
		Blackburn%
		\ \BBA {} Marx%
	}{%
		Blackburn%
		\ \BBA {} Marx%
	}{%
		{\protect \APACyear {2003}}%
	}]{%
		BlMar03}
	\APACinsertmetastar {%
		BlMar03}%
	\begin{APACrefauthors}%
		Blackburn, P.%
		\BCBT {}\ \BBA {} Marx, M.%
	\end{APACrefauthors}%
	\unskip\
	\newblock
	\APACrefYearMonthDay{2003}{}{}.
	\newblock
	{\BBOQ}\APACrefatitle {Constructive interpolation in hybrid logic}
	{Constructive interpolation in hybrid logic}.{\BBCQ}
	\newblock
	\APACjournalVolNumPages{Journal of Symbolic Logic}{68}{2}{463--480}.
	\newblock
	
	\newblock
	
	\newblock
	\begin{APACrefDOI} \doi{10.2178/jsl/1052669059} \end{APACrefDOI}
	\PrintBackRefs{\CurrentBib}
	
	\bibitem [\protect \citeauthoryear {%
		Blackburn%
		\ \BBA {} Tzakova%
	}{%
		Blackburn%
		\ \BBA {} Tzakova%
	}{%
		{\protect \APACyear {1999}}%
	}]{%
		BlTz99}
	\APACinsertmetastar {%
		BlTz99}%
	\begin{APACrefauthors}%
		Blackburn, P.%
		\BCBT {}\ \BBA {} Tzakova, M.%
	\end{APACrefauthors}%
	\unskip\
	\newblock
	\APACrefYearMonthDay{1999}{}{}.
	\newblock
	{\BBOQ}\APACrefatitle {Hybrid languages and temporal logic} {Hybrid languages
		and temporal logic}.{\BBCQ}
	\newblock
	\APACjournalVolNumPages{Logic Journal of the IGPL}{7}{1}{27--54}.
	\newblock
	
	\newblock
	
	\newblock
	\begin{APACrefDOI} \doi{10.1093/jigpal/7.1.27} \end{APACrefDOI}
	\PrintBackRefs{\CurrentBib}
	
	\bibitem [\protect \citeauthoryear {%
		Bra\"uner%
	}{%
		Bra\"uner%
	}{%
		{\protect \APACyear {2011}}%
	}]{%
		Brauner2011}
	\APACinsertmetastar {%
		Brauner2011}%
	\begin{APACrefauthors}%
		Bra\"uner, T.%
	\end{APACrefauthors}%
	\unskip\
	\newblock
	\APACrefYear{2011}.
	\newblock
	\APACrefbtitle {Hybrid Logic and its Proof-Theory} {Hybrid logic and its
		proof-theory}\ (\BVOL~37).
	\newblock
	\APACaddressPublisher{Dordrecht}{Springer}.
	\newblock
	\begin{APACrefDOI} \doi{10.1007/978-94-007-0002-4} \end{APACrefDOI}
	\PrintBackRefs{\CurrentBib}
	
	\bibitem [\protect \citeauthoryear {%
		Fine%
	}{%
		Fine%
	}{%
		{\protect \APACyear {1979}}%
	}]{%
		Fine79}
	\APACinsertmetastar {%
		Fine79}%
	\begin{APACrefauthors}%
		Fine, K.%
	\end{APACrefauthors}%
	\unskip\
	\newblock
	\APACrefYearMonthDay{1979}{}{}.
	\newblock
	{\BBOQ}\APACrefatitle {Failures of the Interpolation Lemma in Quantified Modal
		Logic} {Failures of the interpolation lemma in quantified modal
		logic}.{\BBCQ}
	\newblock
	\APACjournalVolNumPages{Journal of Symbolic Logic}{44}{2}{201--206}.
	\newblock
	
	\newblock
	
	\newblock
	\begin{APACrefDOI} \doi{10.2307/2273727} \end{APACrefDOI}
	\PrintBackRefs{\CurrentBib}
	
	\bibitem [\protect \citeauthoryear {%
		Fitting%
	}{%
		Fitting%
	}{%
		{\protect \APACyear {1975}}%
	}]{%
		Fitting1975}
	\APACinsertmetastar {%
		Fitting1975}%
	\begin{APACrefauthors}%
		Fitting, M.%
	\end{APACrefauthors}%
	\unskip\
	\newblock
	\APACrefYearMonthDay{1975}{}{}.
	\newblock
	{\BBOQ}\APACrefatitle {A Modal Logic Epsilon-Calculus} {A modal logic
		epsilon-calculus}.{\BBCQ}
	\newblock
	\APACjournalVolNumPages{Notre Dame Journal of Formal Logic}{16}{1}{1--16}.
	\newblock
	
	\newblock
	
	\newblock
	\begin{APACrefDOI} \doi{10.1305/ndjfl/1093891609} \end{APACrefDOI}
	\PrintBackRefs{\CurrentBib}
	
	\bibitem [\protect \citeauthoryear {%
		Fitting%
	}{%
		Fitting%
	}{%
		{\protect \APACyear {1996}}%
	}]{%
		Fitting1996}
	\APACinsertmetastar {%
		Fitting1996}%
	\begin{APACrefauthors}%
		Fitting, M.%
	\end{APACrefauthors}%
	\unskip\
	\newblock
	\APACrefYear{1996}.
	\newblock
	\APACrefbtitle {First-Order Logic and Automated Theorem Proving} {First-order
		logic and automated theorem proving}\ (\BVOL~277).
	\newblock
	\APACaddressPublisher{New York}{Springer-Verlag}.
	\newblock
	\begin{APACrefDOI} \doi{10.1007/978-1-4612-2360-3} \end{APACrefDOI}
	\PrintBackRefs{\CurrentBib}
	
	\bibitem [\protect \citeauthoryear {%
		Fitting%
		\ \BBA {} Mendelsohn%
	}{%
		Fitting%
		\ \BBA {} Mendelsohn%
	}{%
		{\protect \APACyear {1998}}%
	}]{%
		FitMen98}
	\APACinsertmetastar {%
		FitMen98}%
	\begin{APACrefauthors}%
		Fitting, M.%
		\BCBT {}\ \BBA {} Mendelsohn, R.L.%
	\end{APACrefauthors}%
	\unskip\
	\newblock
	\APACrefYear{1998}.
	\newblock
	\APACrefbtitle {First-Order Modal Logic} {First-order modal logic}\
	(\BVOL~277).
	\newblock
	\APACaddressPublisher{Dordrecht}{Springer}.
	\newblock
	\begin{APACrefDOI} \doi{10.1007/978-94-011-5292-1} \end{APACrefDOI}
	\PrintBackRefs{\CurrentBib}
	
	\bibitem [\protect \citeauthoryear {%
		Indrzejczak%
	}{%
		Indrzejczak%
	}{%
		{\protect \APACyear {2010}}%
	}]{%
		Indrzejczak2010}
	\APACinsertmetastar {%
		Indrzejczak2010}%
	\begin{APACrefauthors}%
		Indrzejczak, A.%
	\end{APACrefauthors}%
	\unskip\
	\newblock
	\APACrefYear{2010}.
	\newblock
	\APACrefbtitle {Natural Deduction, Hybrid Systems and Modal Logics} {Natural
		deduction, hybrid systems and modal logics}\ (\BVOL~30).
	\newblock
	\APACaddressPublisher{Dordrecht}{Springer}.
	\newblock
	\begin{APACrefDOI} \doi{10.1007/978-90-481-8785-0} \end{APACrefDOI}
	\PrintBackRefs{\CurrentBib}
	
	\bibitem [\protect \citeauthoryear {%
		Indrzejczak%
	}{%
		Indrzejczak%
	}{%
		{\protect \APACyear {2020}}%
	}]{%
		Indrzejczak2020a}
	\APACinsertmetastar {%
		Indrzejczak2020a}%
	\begin{APACrefauthors}%
		Indrzejczak, A.%
	\end{APACrefauthors}%
	\unskip\
	\newblock
	\APACrefYearMonthDay{2020}{}{}.
	\newblock
	{\BBOQ}\APACrefatitle {Existence, definedness and definite descriptions in
		hybrid modal logic} {Existence, definedness and definite descriptions in
		hybrid modal logic}.{\BBCQ}
	\newblock
	N.~Olivetti, R.~Verbrugge, S.~Negri\BCBL {}\ \BBA {} G.~Sandu\ (\BEDS),
	\APACrefbtitle {Advances in {M}odal {L}ogic 13} {Advances in {M}odal {L}ogic
		13}\ (\BPGS\ 349--368).
	\newblock
	\APACaddressPublisher{Rickmansworth}{College Publications}.
	\PrintBackRefs{\CurrentBib}
	
	\bibitem [\protect \citeauthoryear {%
		Indrzejczak%
	}{%
		Indrzejczak%
	}{%
		{\protect \APACyear {2021}}%
	}]{%
		Indrzejczak2021c}
	\APACinsertmetastar {%
		Indrzejczak2021c}%
	\begin{APACrefauthors}%
		Indrzejczak, A.%
	\end{APACrefauthors}%
	\unskip\
	\newblock
	\APACrefYearMonthDay{2021}{}{}.
	\newblock
	{\BBOQ}\APACrefatitle {Russellian definite description theory---{A}
		proof-theoretic approach} {Russellian definite description theory---{A}
		proof-theoretic approach}.{\BBCQ}
	\newblock
	\APACjournalVolNumPages{The Review of Symbolic Logic}{First View}{}{1--26}.
	\newblock
	
	\newblock
	
	\newblock
	\begin{APACrefDOI} \doi{10.1017/S1755020321000289} \end{APACrefDOI}
	\PrintBackRefs{\CurrentBib}
	
	\bibitem [\protect \citeauthoryear {%
		Indrzejczak%
		\ \BBA {} Zawidzki%
	}{%
		Indrzejczak%
		\ \BBA {} Zawidzki%
	}{%
		{\protect \APACyear {2021}}%
	}]{%
		IndZaw}
	\APACinsertmetastar {%
		IndZaw}%
	\begin{APACrefauthors}%
		Indrzejczak, A.%
		\BCBT {}\ \BBA {} Zawidzki, M.%
	\end{APACrefauthors}%
	\unskip\
	\newblock
	\APACrefYearMonthDay{2021}{}{}.
	\newblock
	{\BBOQ}\APACrefatitle {Tableaux for Free Logics with Descriptions} {Tableaux
		for free logics with descriptions}.{\BBCQ}
	\newblock
	\APACrefbtitle {Proceedings of the 30th International Conference on Automated
		Reasoning with Analytic Tableaux and Related Methods} {Proceedings of the
		30th international conference on automated reasoning with analytic tableaux
		and related methods}\ (\BPGS\ 56--73).
	\newblock
	\APACaddressPublisher{Berlin, Heidelberg}{Springer-Verlag}.
	\newblock
	\begin{APACrefDOI} \doi{10.1007/978-3-030-86059-2_4} \end{APACrefDOI}
	\PrintBackRefs{\CurrentBib}
	
	\bibitem [\protect \citeauthoryear {%
		Indrzejczak%
		\ \BBA {} Zawidzki%
	}{%
		Indrzejczak%
		\ \BBA {} Zawidzki%
	}{%
		{\protect \APACyear {2023}}%
	}]{%
		IndZaw2023}
	\APACinsertmetastar {%
		IndZaw2023}%
	\begin{APACrefauthors}%
		Indrzejczak, A.%
		\BCBT {}\ \BBA {} Zawidzki, M.%
	\end{APACrefauthors}%
	\unskip\
	\newblock
	\APACrefYearMonthDay{2023}{}{}.
	\newblock
	{\BBOQ}\APACrefatitle {When Iota Meets Lambda} {When iota meets lambda}.{\BBCQ}
	\newblock
	\APACjournalVolNumPages{Synthese}{201}{}{71}.
	\newblock
	
	\newblock
	
	\newblock
	\begin{APACrefDOI} \doi{10.1007/s11229-023-04048-y} \end{APACrefDOI}
	\PrintBackRefs{\CurrentBib}
	
	\bibitem [\protect \citeauthoryear {%
		Orlandelli%
	}{%
		Orlandelli%
	}{%
		{\protect \APACyear {2021}}%
	}]{%
		Orlandelli2021}
	\APACinsertmetastar {%
		Orlandelli2021}%
	\begin{APACrefauthors}%
		Orlandelli, E.%
	\end{APACrefauthors}%
	\unskip\
	\newblock
	\APACrefYearMonthDay{2021}{}{}.
	\newblock
	{\BBOQ}\APACrefatitle {Labelled calculi for quantified modal logics with
		definite descriptions} {Labelled calculi for quantified modal logics with
		definite descriptions}.{\BBCQ}
	\newblock
	\APACjournalVolNumPages{Journal of Logic and Computation}{31}{3}{923--946}.
	\newblock
	
	\newblock
	
	\newblock
	\begin{APACrefDOI} \doi{10.1093/logcom/exab018} \end{APACrefDOI}
	\PrintBackRefs{\CurrentBib}
	
	\bibitem [\protect \citeauthoryear {%
		Smullyan%
	}{%
		Smullyan%
	}{%
		{\protect \APACyear {1968}}%
	}]{%
		Smullyan1968}
	\APACinsertmetastar {%
		Smullyan1968}%
	\begin{APACrefauthors}%
		Smullyan, R.M.%
	\end{APACrefauthors}%
	\unskip\
	\newblock
	\APACrefYear{1968}.
	\newblock
	\APACrefbtitle {First-Order Logic} {First-order logic}.
	\newblock
	\APACaddressPublisher{Berlin, Heidelberg}{Springer Verlag}.
	\newblock
	\begin{APACrefDOI} \doi{10.1007/978-3-642-86718-7} \end{APACrefDOI}
	\PrintBackRefs{\CurrentBib}
	
	\bibitem [\protect \citeauthoryear {%
		Stalnaker%
		\ \BBA {} Thomason%
	}{%
		Stalnaker%
		\ \BBA {} Thomason%
	}{%
		{\protect \APACyear {1968}}%
	}]{%
		StalnakerThomason1968}
	\APACinsertmetastar {%
		StalnakerThomason1968}%
	\begin{APACrefauthors}%
		Stalnaker, R.C.%
		\BCBT {}\ \BBA {} Thomason, R.H.%
	\end{APACrefauthors}%
	\unskip\
	\newblock
	\APACrefYearMonthDay{1968}{}{}.
	\newblock
	{\BBOQ}\APACrefatitle {Abstraction in First-Order Modal Logic1} {Abstraction in
		first-order modal logic1}.{\BBCQ}
	\newblock
	\APACjournalVolNumPages{Theoria}{34}{3}{203--207}.
	\newblock
	
	\newblock
	
	\newblock
	\begin{APACrefDOI} \doi{10.1111/j.1755-2567.1968.tb00351.x} \end{APACrefDOI}
	\PrintBackRefs{\CurrentBib}
	
	\bibitem [\protect \citeauthoryear {%
		Zawidzki%
	}{%
		Zawidzki%
	}{%
		{\protect \APACyear {2014}}%
	}]{%
		Zawidzki2014}
	\APACinsertmetastar {%
		Zawidzki2014}%
	\begin{APACrefauthors}%
		Zawidzki, M.%
	\end{APACrefauthors}%
	\unskip\
	\newblock
	\APACrefYear{2014}.
	\newblock
	\APACrefbtitle {Deductive Systems and the Decidability Problem for Hybrid
		Logics} {Deductive systems and the decidability problem for hybrid logics}.
	\newblock
	\APACaddressPublisher{\L\'od\'z, Krak\'ow}{Lodz University Press/Jagiellonian
		University Press}.
	\PrintBackRefs{\CurrentBib}
	
\end{thebibliography}
\end{document}